\documentclass[11pt,letter]{article}

\usepackage[T1]{fontenc}

\usepackage{mdframed}
\usepackage{lipsum}
\usepackage{amsmath,amsthm,amssymb,amsfonts,bm,nicefrac,algpseudocode}
\usepackage{mathrsfs}
\usepackage{thmtools}
\usepackage{thm-restate}
\usepackage{lineno}
\usepackage{authblk}

\allowdisplaybreaks

\usepackage{fullpage}
\usepackage[margin = 2.54cm]{geometry}

\usepackage[numbers]{natbib}

\usepackage{epsfig}
\usepackage{microtype}
\usepackage{wrapfig}
\usepackage{graphicx}
\usepackage{enumitem}
\usepackage{indentfirst,relsize}
\usepackage{setspace}
\usepackage{latexsym}
\usepackage{stackrel}
\usepackage[all]{xy}
\usepackage[usenames,dvipsnames]{pstricks}
\usepackage{pst-grad} %
\usepackage{pst-plot} %
\usepackage{xspace}
\usepackage{bbm}
\usepackage{cancel}
\usepackage{url}

\usepackage[super]{nth}
\usepackage[colorlinks=true,urlcolor=blue,linkcolor=blue,citecolor={red!75!black},unicode, pagebackref, breaklinks]{hyperref}

\usepackage[nameinlink]{cleveref}

\usepackage[ruled,linesnumbered,vlined,boxed,commentsnumbered]{algorithm2e}

\SetKwProg{Procedure}{Procedure}{}{}
\SetKwInOut{Given}{Given}

\newcommand{\size}[1]{\left| #1 \right|}

\newcommand{\E}{\mathop{\mathbb{E}}}
\newcommand{\remove}[1]{}

\newcommand{\F}{\mathbb{F}}
\newcommand{\R}{\mathbb{R}}

\newcommand{\N}{\mathbb{N}}

\newcommand{\indic}{\mathbbm{1}}
\newcommand{\Z}{\mathbb{Z}}
\newcommand{\caD}{\mathcal{D}}
\DeclareMathOperator{\pline}{L}
\newcommand{\caP}{\mathcal{P}}

\newcommand{\cI}{\mathcal{I}}

\newcommand{\cB}{\mathcal{B}}
\newcommand{\cC}{\mathcal{C}}
\newcommand{\cD}{\mathcal{D}}

\newcommand{\cF}{\mathcal{F}}
\newcommand{\cP}{\mathcal{P}}
\newcommand{\cN}{\mathcal{N}}

\newcommand{\cR}{\mathcal{R}}

\newcommand{\lrad}{R}
\newcommand{\srad}{r}

\newcommand{\var}{\mathop{\mathsf{Var}}}
\newcommand{\maxop}{\mathop{\text{max}}}
\newcommand{\minop}{\mathop{\text{min}}}

\newcommand{\cov}{\mathop{\sf Cov}}
\newcommand{\eps}{\varepsilon}
\newcommand{\pr}{\mathrm{Pr}}

\newcommand{\ceil}[1]{{\lceil{#1}\rceil}}
\newcommand{\floor}[1]{{\lfloor{#1}\rfloor}}

\renewcommand{\vec}[1]{\ensuremath{\boldsymbol{#1}}}
\newcommand{\iprod}[2]{\ensuremath{\left\langle#1,#2\right\rangle}}
\DeclareMathOperator{\ball}{B}

\newcommand{\ideg}{\mathrm{ideg}}
\newcommand{\dist}{\mathsf{dist}} %
\newcommand{\ignore}[1]{} %

\remove{\usepackage[symbol]{footmisc}

}

\theoremstyle{plain}
\newtheorem{theo}{Theorem}[section]

\newtheorem{rem}[theo]{Remark}
\newtheorem{obs}[theo]{Observation}
\newtheorem{lem}[theo]{Lemma}

\newtheorem{cl}[theo]{Claim}

\theoremstyle{definition}
\newtheorem{defi}[theo]{Definition}
\newtheorem{fact}[theo]{Fact}

\usepackage{cleveref}

\usepackage{todonotes}

\newcounter{mynotes}
\newcommand{\paren}[1]{\left(#1\right)}

\newcommand{\infl}{\mathsf{Infl}\xspace}

\newcommand{\testlin}{\textsc{Test-$k$-Linear}\xspace}

\newcommand{\testjunta}{\textsc{Test-$k$-Junta}\xspace}

\newcommand{\testsparse}{\textsc{Test-$k$-Sparse}\xspace}

\newcommand{\apxqueryg}{\textsc{ApproxQuery-}g\xspace}

\newcommand{\apxlowdegtest}{\textsc{ApproxLowDegreeTester}\xspace}

\newcommand{\findbucket}{\mathsf{FindInfBucket}\xspace}

\newcommand{\findbuckets}{\mathsf{FindInfBuckets}\xspace}

\title{Testing Sparse Functions over the Reals}

\author[1]{Vipul Arora\thanks{This work was done in part while the author was visiting the Simons Institute for the Theory of Computing.}}
\affil{\footnotesize National University of Singapore. \texttt{\href{mailto:vipul@comp.nus.edu.sg}{vipul@comp.nus.edu.sg}}.}
\author[2]{Arnab Bhattacharyya\thanks{Research supported by a start-up
grant at the University of Warwick.}}
\affil{\footnotesize University of Warwick. \texttt{\href{mailto:arnab.bhattacharyya@warwick.ac.uk}{arnab.bhattacharyya@warwick.ac.uk}}.}
\author[3]{Philips George John}
\affil{\footnotesize CNRS@CREATE \& NUS. \texttt{\href{mailto:philips.george.john@u.nus.edu}{philips.george.john@u.nus.edu}}.}
\author[3]{Sayantan Sen\thanks{Research supported by the NRF Investigatorship award (NRF-NRFI10-2024-0006)
and CQT Young Researcher Career Development Grant (25-YRCDG-SS).}}
\affil{\footnotesize Centre for Quantum Technologies, National University of Singapore. \texttt{\href{mailto:sayantan789@gmail.com}{sayantan789@gmail.com}}.}

\date{
}

\begin{document}

\maketitle

\begin{abstract}
Over the last three decades, function testing has been extensively studied over Boolean, finite fields, and discrete settings. However, to encode the real-world applications more succinctly, function testing over the reals (where the domain and range, both are reals) is of prime importance. Recently, there have been some works in the direction of testing for algebraic representations of such functions: the work by Fleming and Yoshida (ITCS 20), Arora, Kelman, and Meir (SOSA 25) on \emph{linearity testing} and the work of Arora, Bhattacharyya, Fleming, Kelman, and Yoshida (SODA 23) for testing \emph{low-degree polynomials}. Our work follows the same avenue, wherein we study three well-studied sparse representations of functions, over the reals, namely (i) \emph{$k$-linearity}, (ii) \emph{$k$-sparse polynomials}, and (iii) \emph{$k$-junta}. 

In this setting, given approximate query access to some $f:\R^n \rightarrow \R$, we want to decide if the function satisfies some property of interest, or if it is \emph{far} from all functions that satisfy the property.
Here, the distance is measured in the $\ell_1$-metric, under the assumption that we are drawing samples from the Standard Gaussian distribution. We present efficient testers and $\Omega(k)$ lower bounds for testing each of these three properties.

\end{abstract}

\section{Introduction}
Property testing~\cite{blum1990self,RubinfeldS96,GGR} is a rigorous framework for studying the global properties of large datasets by accessing only a few entries of it. 
In particular, given query access to an unknown ``huge object'', the goal of property testing is to verify some property of the object by only inspecting a small portion of it.
Formally, we can define property testing of functions, or \emph{function testing}, by considering the ``huge object'' as a function $f$ over an underlying domain, which we can \emph{query} at a small number of points in order to verify a property of the function $f$ in an approximate sense. 
For example, consider the problem of \emph{linearity testing} of functions $f:\F^n \rightarrow \F$, where $\F$ is a finite field. We are given access to a query oracle for $f$; i.e., on input $\bm{x}$, the oracle returns $f(\bm{x})$.
The goal is to distinguish with high probability between two cases, viz. (i) $f$ is a linear function, or (ii) $f$ is ``far'' from all linear functions, while performing as few queries to the oracle as possible.

The field of property testing was initiated in the seminal work of \cite{blum1990self,blum1993self}, who studied the problem of \emph{self-testing} of programs, where the goal was to understand the correctness of a program by verifying its outputs on a set of correlated input data. 
A fundamental problem they studied is that of testing whether $f$ is \emph{linear} (or generally, a homomorphism of abelian groups) or $\varepsilon$-far from linearity, with the notion of $\eps$-farness of $f: \F^n \rightarrow \F$ from a function class $\cP$ defined as: 
\[\mathsf{dist}_{\mathsf{Unif}(\F^n),\ell_0}(f,\cP) = \inf_{g \in \cP} \Pr_{\bm x \sim \mathsf{Unif}(\F^n)}[f({\bm x}) \neq g({\bm x})] \geq \eps,\]
i.e., for any function $g$ that satisfies property $\cP$, $f$ disagrees with $g$ on at least an $\eps$-fraction of the inputs $\bm x$ (uniformly drawn from $\F^n$).
In the context of linearity testing, $\cP$ is the class of linear functions/homomorphisms (the $\ell_0$-distance definition given above can be easily generalized in terms of function domain and range, input distribution etc). They showed a constant query (independent of $|\F|$ and $n$) linearity tester, eponymously known as the \emph{BLR tester}.

Over the last three decades, property testing has been extensively studied in many settings, such as when the unknown object is a function, a graph, or a probability distribution, with many natural connections to real-world problems, e.g., in the context of probabilistically checkable proofs~\cite{pcp1,arora-safra,raz1997sub,dinur}, PAC learning~\cite{GGR}, program checking~\cite{RubinfeldS92,RubinfeldS96}, approximation algorithms~\cite{GGR,pcp2} and many more.

In particular, starting from the very first work that initiated the field of property testing~\cite{blum1990self}, 
over the last three decades, function property testing has been extensively studied in several settings, culminating in a wide array of tools and techniques. This includes settings where the function is defined over finite fields~\cite{alon2005testing,kaufman2006testing,jutla2009testing,friedl1995some,bhattacharyya2010optimal,bhattacharyya2013every,gopalan2011testing,samorodnitsky2007low}, over hyper-grids~\cite{dodis1999improved,blais2014lower,chakrabarty2013optimal,baleshzar2017optimal}, etc. One may see the books~\cite{goldreich2017introduction, bhattacharyya2022property} and the surveys~\cite{DBLP:journals/eatcs/Fischer01,ron2008property,DBLP:journals/fttcs/Ron09} for detailed references. 

\paragraph*{Our problem setting} Suppose we are given query access to an unknown function $f: \R^n \rightarrow \R$, and 
our goal is to distinguish whether (i) $f$ satisfies some property $\cP$, or (ii) $f$ is $\eps$-far from all functions that satisfy $\cP$. We define the notion of $\eps$-farness over the reals in the following way: we fix a reference distribution $\cD$ on $\R^n$, and we say $f$ is $\eps$-far from $\cP$ if the following holds\footnote{The notion of $\ell_1$-distance is more appropriate for real-valued functions, rather than the commonly used Hamming ($\ell_0$) distance, especially in the \emph{approximate query setting} which we define later. For example, consider the case if $g = f + \varepsilon$ for some $\varepsilon > 0$, then $\mathsf{dist}_{\cD,\ell_1}(f,g) = \varepsilon$, whereas $\mathsf{dist}_{\cD,\ell_0}(f,g) = 1$.
}:
\[\mathsf{dist}_{\cD,\ell_1}(f,\cP) \triangleq \inf_{g \in \cP} \left\{\E_{\bm x \sim \cD}\,[|f(\bm x) - g(\bm x)|]\right\} \geq \eps.\]

Studying function testing over the reals often requires new techniques compared to testing on finite domains. A common reference distribution on $\R^n$ is the standard $n$-dimensional Gaussian distribution $\cN(\bm 0,I_n)$, which is approximately a uniform distribution on an $\ell_2$-sphere of radius $\sqrt{n}$. And the notion of, exact queries (oracle giving $f(\bm x)$) is not realistic over the reals since one may require exponentially many bits to represent the exact function value. Although the problem of function testing over the reals has several practical motivations, there have been comparatively few works in this setting till now (compared to testing over finite domains) such as testing surface areas~\cite{neeman2014testing,kothari2014testing}, testing  halfspaces~\cite{matulef2010testing1,matulef2010testing2,matulef2009testing3}, linear separators~\cite{balcan2012active}, high-dimensional convexity~\cite{FSS17}, linear $k$-juntas~\cite{de2019your} etc.

Interestingly, these works mostly focus on the setting when $f$ is Boolean valued, i.e., $f:\R^n \rightarrow \{\pm 1\}$.  
The setting where the range of $f$ is real has also been studied, e.g., in the work of \cite{lptesting}, wherein they test properties of functions, defined over finite hyper-grids, with respect to $L_p$-distances.
In \cite{black_et_al}, functions over the hypercube, i.e., $f:\{0,1\}^n\to\R$ are studied for monotonicity.  Recently, \cite{ferreirapintojr23,Pinto24} studied $L^p$ testing of monotonicity of Lipschitz functions $f:[0,1]^n\to\R$.

\cite{fleming2020distribution} studied the problem of linearity testing for real-valued real-domain functions in full generality, i.e., for functions $f:\R^n\to\R$. Later, \cite{arora2} studied the problem of testing low-degree polynomials in this regime, followed by the work of \cite{arora2024optimaltestinglinearitysosa}, which improved these results by achieving query-optimality with respect to the proximity parameter $\eps$. In this work, we expand the landscape of testing various fundamental notions of sparsity to the general $f: \R^n \to \R$ regime. 

A recent work of \cite{bao2025testingnoisylowdegreepolynomials} studies a problem similar to sparsity testing of low degree polynomials, in a sample-based access model. However, it assumes a promise that the unknown function $f:\R^n\to\R$ is a \emph{multi-linear} polynomial to begin with, and needs a large gap for the sparsity parameter between the YES and NO cases.

When the function $f$ to be tested is real valued, two different models can be considered with respect to the \emph{query accuracy}. 
In the first model, known as the \emph{arbitrary precision arithmetic} or the \emph{exact testing} model, we assume that the oracle representing $f$ can give the exact value of $f$ at any point of query ${\bm x}$. However, this model is unrealistic from an implementation, or even from a classical complexity-theoretic viewpoint. 
So, we instead consider a \emph{finite precision arithmetic} or \emph{approximate testing} model:
\begin{restatable}[$\eta$-\emph{approximate query}]{defi}{etaapproxquery}\label{defi:etaapproxquery}
    The oracle, when queried for $f(\bm x)$, outputs  $\widetilde{f}({\bm x})$ such that $|\widetilde{f}({\bm x}) - f({\bm x})| \leq \eta$, for some small parameter $\eta \in (0,1)$, for every query point ${\bm x}$.
\end{restatable}

It is clear that for any property $\cP$, any tester for $\cP$ in the approximate model will also work in the exact model. 
The notion of approximate testing was studied in earlier works~\cite{gemmell1991self,ar1993checking,ergun2001checking}. 
In this context, $\eta$ can be thought of as the resolution limit of the computational machine, i.e., if the machine offers some $\alpha$ bits of precision, then $\eta= 2^{-\alpha}$. 
$\eta$ can also be thought of as the noise reliability threshold of a channel communicating reals, i.e., if some information $a \in \R$ is transmitted on a channel with a reliability threshold of $\eta$, then the received observable $\widetilde a\in\R$ satisfies $|a-\widetilde a|\leq\eta$. 
When $\eta = 0$, this is the exact query model. 

All our testers in this work are analyzed in the approximate query model for $\eta > 0$. Our results also require appropriate upper bounds on $\eta$, which will be stated in each context. 
But we use the exact query model for proving the lower bounds, since lower bounds for $\eta=0$ also hold for any $\eta$-approximate query model for $\eta>0$.

Similar to the context of access to the unknown functions, there are some variations in terms of the error profile of the testing algorithms. 
A tester is said to have \emph{two-sided error} if it can err in both the cases: when $f \in \cP$, and when $f$ is $\eps$-far from $\cP$. 
This is in contrast with \emph{one-sided error} testers, which always decide correctly when $f \in \cP$, and can only err when $f$ is $\eps$-far from $\cP$.
Likewise, a tester is said to be \emph{adaptive} if it performs queries based on the answers it obtained for the previous queries. On the other hand, a \emph{non-adaptive} tester performs all its queries together in a single round.

A tester over a finite domain is said to be \emph{local} if the number of queries performed by it is independent of the domain size.

In this work, we design testers in the approximate query model for three sparse function representations: $k$-linear, $k$-sparse low-degree polynomials, and $k$-juntas. 
Our testers for $k$-linearity, and $k$-sparse low-degree polynomials have two-sided error, whereas our $k$-junta tester has one-sided error. 
Our testers for $k$-linear functions and $k$-juntas are adaptive, while our tester for $k$-sparse low-degree polynomials is non-adaptive. \textcolor{black}{Importantly, all our testers are \emph{local}: the number of queries performed by our testers is independent of the domain dimension $n$, and depends only on the sparsity parameter $k$ and the proximity parameter $\varepsilon$ (as well as the total degree $d$ in the $k$-sparse low-degree polynomial case)}.

These properties are well-studied in the Boolean and finite fields regime, and this work strives to do the same over continuous domains. We believe the new techniques we have developed to design these testers will be of independent interest.

{
\color{black}
\paragraph*{Boundedness} We assume that the unknown functions are bounded inside an $\ell_2$ ball of suitable radius (typically $O(\sqrt{n})$), i.e., for every $\bm x\in \ball(\bm 0,O(\sqrt{n}))$, $f(\bm x) \leq C \|\bm x\|_2$ for some fixed constant $C$ (so $f(\bm x) \leq C \sqrt{n}$ in $\ball(\bm 0, 2\sqrt{n})$)~\footnote{This boundedness notion corresponds to the bounds we get for low-degree polynomials over compact domains, e.g., if $f(x) = a_1 x + \ldots + a_d x^d$ is a polynomial in $x$, then
$|f(x)| \leq \sum_{i=1}^{d} |a_i| R^i$ for any $x \in [-R, R]$, for all $R\geq 0$.}. Note that testing only via bounded queries is impossible without such an assumption. For example, consider an arbitrary ``good'' function $f: \R^n \rightarrow \R$ ($k$-linear, $k$-sparse polynomial, or $k$-junta), and choose a suitably small region $R$ at random (e.g. by choosing ${\bm y}$ uniformly from $\mathrm{B}(\bm{0}, \sqrt{n})$ and setting $R$ to be a small radius ball around ${\bm y}$), where $0 < \mu_n(R) \ll 1/n^c$ for any $c > 0$ under the $\cN(\bm 0, I_n)$ Gaussian measure on $\R^n$. Further, subdivide $R$ equally into $R_1$ and $R_2$, and for some $A \geq 2\eps/\mu_n(R)$, define a function $f^\prime:\R^n\to\R$ as:
\[f^\prime(\bm x)\triangleq\begin{cases} f(\bm x), &\text{ if }\bm x\not\in R,\\
A, &\text{ if }\bm x\in R_1,\text{ and }\\
-A, &\text{ if }\bm x\in R_2\end{cases}.\]
As $A$ can be arbitrarily large, $f^\prime$ can be moved arbitrarily far from the required property (in $\ell_1$ distance over the Gaussian measure). However, since $\mu_n(R)$ is arbitrarily small, no algorithm using bounded queries can distinguish between $f$ and $f^\prime$.

\paragraph*{Choice of Reference Distribution}
In this work (for the $k$-linear and $k$-sparse low-degree testers), we have used an anti-concentration result of Glazer and Mikulincer (\Cref{srt_thm:glazer-mikulincer}) which uses the Carbery-Wright anti-concentration inequality (\autoref{srt_thoe:carbery-wright}), but with a more usable expression for the variance (in terms of the polynomial coefficients). The form that we have used assumes that the distribution (on the variables) is log-concave, isotropic, and non-discrete. We have chosen $\cN(\bm 0, I_n)$ as representative of such a distribution, but we can also take $P^{\otimes n}$ for a continuous log-concave distribution $P$ on $\R$. A probabilistic upper bound for the Hankel matrix eigenvalues, which we use in the analysis of the $k$-sparsity tester, also assumes the isotropic Gaussian distribution (\Cref{lem:upperboundonsigmamax}). 

}

\section{Our results}\label{sec:results}

In this work, we focus on three problems: testing (i) $k$-linearity, (ii) $k$-sparse low-degree polynomials, and (iii) $k$-juntas. Throughout this work, we assume that our reference distribution $\cD$ is the standard Gaussian $\cN(\bm 0, I_n)$, unless otherwise stated.

\subsection{Testing \texorpdfstring{$k$}{k}-linear functions}\label{sec:intro-k-lin}

\begin{defi}[$k$-linearity]\label{def:k-lin}
Let $f: \R^n \rightarrow \R$, and $k \in \N$ be a parameter. $f$ is a \emph{$k$-linear} function if there exists a set $S \subseteq [n]:|S|\leq k$, and there exist $\{c_i\in\R:i\in S\}$, such that for any ${\bm x} \in \R^n$,
\[f({\bm x}) = \sum_{i \in S} c_i{\bm x}_i.\]
\end{defi}

This problem has been extensively studied over finite domains with exact query access. \cite{fischer2004testing} designed the first tester for $k$-linearity with query complexity $\widetilde{O}(k^2)$ by studying the related problem of testing $k$-juntas. 
Later, \cite{blais2009testing} improved the bound for testing $k$-juntas to $O(k \log k)$ queries. 
Using the BLR test~\cite{blum1990self}, along with this result gives a tester for $k$-linearity with $O(k \log k)$ query complexity, as done by \cite{bshouty2023optimal}, who presented an optimal, two-sided error, non-adaptive algorithm for this problem over Boolean domains.
The first lower bounds for this problem were presented by \cite{fischer2004testing}, proving $\Omega(\sqrt{k})$ non-adaptive, and $\Omega(\log k)$ adaptive queries are necessary for testing $k$-linearity. 
These were first improved by \cite{goldreich2010testing}, to $\Omega(k)$ non-adaptive, and $\Omega(\sqrt{k})$ adaptive query lower bounds. This was further improved by \cite{blais2012property,blais2012tight} who proved $\Omega(k)$ adaptive query lower bound. Interestingly, the $\Omega(k)$ adaptive query lower bound by \cite{blais2012property} was proved by showing a novel connection to communication complexity, which we will also use later to prove our lower bounds.
Our result for testing $k$-linearity is summarized in the following theorem.

\begin{theo}[Informal, see \autoref{srt_theo:k_linearity}]\label{srt_thm:k-lin-intro}
Suppose \textcolor{black}{$f : \R^n \rightarrow \R$ is a function bounded in the
ball $B(\bm 0, 2 \sqrt{n})$ and we are given $\eta$-approximate query access to $f$}. Let $k$ be a positive integer and $\eps, \eta \in (0,2/3)$ such that $\eta < \min\left\{\eps, O\left(\min_{\substack{i\in[n]:f(\bm e_i)\neq 0}}\frac{|f(\bm e_i)|}{(nk)^2}\right)\right\}$, where $\bm e_i$ denotes the $i^{th}$ standard basis vector. There exists an $\widetilde{O}(k \log k + \nicefrac{1}{\varepsilon})$-query tester (\autoref{srt_balg:testklinear}) that distinguishes whether $f$ is $k$-linear, or is $\eps$-far from all $k$-linear functions, with probability at least $2/3$.
\end{theo}

\begin{rem}\label{rem:eta-upper-bound-discussion}
    Note the necessity of sufficiently good machine precision (small $\eta$) to successfully test $k$-linearity. This is inevitable because if a linear function has $k$ ``large'' coefficients $a_1,\ldots,a_k$ ($\geq \eta$) and a small coefficient $a_{k+1}$ (say $<\eta/n^2$), we cannot distinguish between such a function and a $k$-linear function $\sum_{i=1}^{k} a_i x_i$ using $\eta$-approximate queries and bounding the absolute difference between the function values (as we do in \autoref{srt_balg:testklinear}), since the function values will differ by at most $|a_{k+1}(x_{k+1})|$. Similar restrictions apply to testing $k$-sparse polynomials, and $k$-juntas.

    In another sense, a qualitative dependence of $\eta$ on the function structure is inevitable. Otherwise, if given an $\eta$-approximate oracle  $\widetilde{f}$ to $f$, $\widetilde{f}^{\prime} \triangleq \widetilde{f}/2^n$ (which can be computed in $O(n)$ time given $\widetilde{f}$, assuming a variable-length binary floating-point representation) would be an $\eta/2^n$-approximate oracle to $f$. But such a scaling would not help with our results, since the coefficients of $f$, the non-zero influences, etc., would be similarly scaled-down.
\end{rem}

\subsection{Testing \texorpdfstring{$k$}{k}-sparse low degree polynomials}\label{sec:k-sparse-intro}

\begin{defi}[$k$-sparsity]\label{def:k-sparse}
Let $f: \R^n \rightarrow \R$, and $k \in \N_{>0}$ be a parameter. 
Define $\bm x^{\bm\alpha}\triangleq\prod_{i=1}^n x_i^{\alpha_i}$, 
for any ${\bm x}\triangleq (x_1, \ldots, x_n)\in\R^n$, 
and $\bm\alpha\triangleq(\alpha_1,\hdots,\alpha_n)\in\N^n$. 
A polynomial $f(x_1, \ldots, x_n) = \sum_{i=1}^{\ell} a_i {\bm x}^{\bm d_i}$, with $a_i \neq 0$, and $\bm d_i\in\N^n$ for every $i \in [\ell]$, is said to be a \emph{$k$-sparse} polynomial if $\ell \leq k$.    
\end{defi}

Grigorescu \emph{et al.}~\cite{gjr10} studied this problem for functions on finite fields, i.e., $f:\F_q^n\to\F_q$, assuming that $q$ is large enough, using the machinery of Hankel matrices associated with a polynomial:

\begin{restatable}[Hankel Matrix for polynomials \cite{gjr10,bot88}]{defi}{Hankelmatrixpoly}\label{srt_defi:Hankel-matrix-poly}
    Consider any $\bm u\triangleq (u_1,\hdots,u_n)\in\R^n$, and define $\bm u^i\triangleq (u_1^i,\hdots,u_n^i)\in\R^n,\forall i\in\N$. For a function $f:\R^n\to\R$ and an integer $t\in\Z_{>0}$, define the $t$-dimensional \emph{Hankel matrix} associated with $f$ at $\bm u$ to be the following:
	\[ H_t(f,\bm u)\triangleq
	\begin{pmatrix}
		f(\bm u^0) & f(\bm u^1) & \hdots & f(\bm u^{t-1}) \\
		f(\bm u^1) & f(\bm u^2) & \hdots & f(\bm u^t) \\
		\vdots & \vdots & \ddots & \vdots \\
		f(\bm u^{t-1}) & f(\bm u^t) & \hdots & f(\bm u^{2t-2})
	\end{pmatrix} \in \R^{t \times t}.
	\]
\end{restatable}
From Ben-Or and Tiwari~\cite{bot88}'s \autoref{obs:Hankel-charcterization}, they designed a tester with a query complexity of $O(k)$ (independent of $d$), assuming $f$ to be an individual-degree-$d$ polynomial (note that all functions $\F_q^n \rightarrow \F_q$ are polynomials of individual degree $\leq q$).

Note that if we have exact query access to $f$, the Hankel matrix $H_t(f,\bm u)$ can be computed using only $2t-1$ queries to $f$ for any point $\bm u\in\R^n$.
For the problem of testing a \emph{polynomial} $f:\R^n\to\R$ for sparsity, given exact query access, the tester of \cite{gjr10} works as it is. In fact, it achieves perfect soundness and completeness (proved in \autoref{srt_lem:perfect-sparsity-tester}). However, the assumption of an exact query oracle is not realistic.
\textcolor{black}{Moreover, for general functions $f: \R^n \to \R$, the promise of polynomiality no longer holds}, and therefore a preliminary step is needed to eliminate functions that are far from being low-degree polynomials. \cite{arora2} designed a \emph{local}, approximate query, low degree tester with a query complexity of $O(d^5)$, for polynomials over $\R^n$ with \emph{total degree} $d$, which we will use in our work. Thus, we restrict our attention to polynomials of \emph{total degree} at most $d$.

\begin{theo}[Informal, see \autoref{srt_theo:ksparsefinalmain}]\label{srt_thm:k-sparse-intro}
Let $k,d \in \N$, $\eps, \eta \in (0,1)$ be parameters such that $\eta < \min\{\eps, 1/2^{2^n}\}$, and \textcolor{black}{$f: \R^n \rightarrow \R$ be bounded in $\ball(\bm 0,2d\sqrt{n})$}, given via an $\eta$-approximate query access. Then there exists an $\widetilde{O}(d^5 + \nicefrac{d^2}{\eps} + d k^3)$-query tester (\autoref{srt_alg:testksparse}), that distinguishes whether $f$ is a $k$-sparse,  degree-$d$ polynomial, or is $\eps$-far from all such polynomials, with probability $\geq 2/3$.
\end{theo}

\begin{rem}\label{rem:compare-k-lin-sparse}
We note that \testsparse (\autoref{srt_alg:testksparse}) also works for testing $k$-linear functions (by setting the degree $d=1$).
However, in this context, we make the following remarks.
\begin{itemize}
    \item[(i)] Since the query complexity of \testsparse is $\widetilde{O}(d^5 + \nicefrac{d^2}{\eps} + d k^3)$, if we invoke it for $k$-linearity testing, the bound would be $\widetilde{O}(k^3 + \nicefrac{1}{\eps})$ which is worse than that of \testlin.

    \item[(ii)] The restriction on the approximation parameter $\eta$ is $1/2^{2^n}$ for \testsparse, as compared to $1/(nk)^2$ for \testlin. Thus, for a wider range of parameters when $\eta$ is not too small, invoking \testlin is better, compared to \testsparse.

    \item[(iii)] \testlin is an adaptive tester, while \testsparse is non-adaptive.
\end{itemize}
\end{rem}

\subsection{Testing \texorpdfstring{$k$}{k}-juntas}\label{sec:intro-k-junta}

\begin{defi}[$k$-junta]\label{def:k-junta}
Let $f: \R^n \rightarrow \R$, and $k \in \N$ be a parameter. A coordinate $i \in [n]$ is said to be \emph{influential} with respect to $f$, if for some $\bm x\in\R^n$, changing the value of ${\bm x}_i$ changes the value of $f(\bm x)$. $f$ is said to be a \emph{$k$-junta}, if there are at most $k$ influential variables with respect to $f$.    
\end{defi}

Testing whether a Boolean function is a $k$-junta has been extensively studied in the exact query model.  The first result in this context was by \cite{parnas2002testing}, which was followed by the work of \cite{fischer2004testing}, who designed a $\widetilde{O}(k^2)$-query tester. Later, \cite{diakonikolas2007testing} extended it to the finite range setting.  \cite{blais2008improved} then gave an $\widetilde{O}(k^{3/2})$-query non-adaptive tester for this problem, while for adaptive testers, \cite{blais2009testing} showed $\widetilde{O}(k \log k + k/\eps)$ queries suffice. It is important to note that all these results use Fourier-analytic techniques.

Notably, \cite{blais2015partially} designed a new algorithm for testing $k$-juntas with similar optimal bounds in the context of testing partial isomorphism of functions.  Interestingly, this work deviates from the common Fourier analytic approach and instead presents a combinatorial approach to this problem. This algorithm from \cite{blais2015partially}  will be used in designing our tester for $k$-juntas. Recently \cite{de2019your} studied the linear $k$-junta testing problem, where the function $f$ is defined as $f: \R^n \rightarrow \{+1,-1\}$~\footnote{A function $f: \R^n \rightarrow \{-1, 1\}$ is said to be a \emph{linear k-junta} if there are $k$ unit vectors $u_1 \ldots,  u_k \in \R^n$ and $g: \R^k \rightarrow \{-1, 1\}$ such that $f(x) =
g(\langle u_1, x \rangle, \ldots, \langle u_k , x \rangle)$.}. This is different from our setting of real-valued functions. 
It is not clear to us if their techniques can be generalized to our setting.

In terms of lower bounds, \cite{fischer2004testing} showed a lower bound of $\Omega(\sqrt{k})$ queries for non-adaptive testers, which was improved to $\widetilde{\Omega}(k^{3/2}/\eps)$ by \cite{ChenSTWX18}. For adaptive testers, \cite{chockler2004lower} showed an $\Omega(k)$ lower bound, which was then improved to $\Omega(k\log k)$ by \cite{Saglam18}.

Our result for testing $k$-juntas is summarized in the following theorem.

\begin{theo}[Informal, see \autoref{theo:kjuntafinalmain}]\label{srt_thm:junta-intro}
Let $k \in \N$, and $\eps, \eta \in (0,1)$ be parameters such that $\eta < \min\{O(\eps/k^2), O(\nicefrac{1}{k^2 \log^2 k})\}$, and \textcolor{black}{$f: \R^n \rightarrow \R$ is bounded in $\ball(\bm 0,2\sqrt{n})$} given via $\eta$-approximate query access. There exists a one-sided error $\widetilde{O}\left(\frac{k \log k}{\eps}\right)$-query tester (\autoref{alg:testjuntaapprox}), that distinguishes if $f$ is a $k$-junta, or is $\eps$-far from all $k$-juntas, with probability at least $2/3$.

\end{theo}

\subsection{Lower bounds}\label{sec:intro-lower-bounds}

Now we briefly mention our lower bound results, which hold even for adaptive testers. For all these three properties ($k$-linearity, $k$-sparse degree-$d$ polynomials, and $k$-juntas), we prove lower bounds of $\Omega(\max\{k, \nicefrac{1}{\eps}\})$ queries. Additionally, for $k$-sparse degree-$d$ polynomials, we prove a lower bound of $\Omega(\max\{k, d, \nicefrac{1}{\eps}\})$ queries. 
All our lower bounds follow from the general reduction from communication complexity, introduced by \cite{blais2012property}, coupled with some folklore results. 
We use \textsc{Set-Disjointness} as the hard instance to prove our lower bounds.

\begin{restatable}{theo}{lowerbounds}\label{srt_thm:lowerbounds}
Given exact query access to $f: \R^n \rightarrow \R$, some $k, d \in\N$ and a distance parameter $\eps \in (0,1)$, $\Omega\left(\max\{k, \nicefrac{1}{\eps}\}\right)$ queries are necessary for testing the following properties with probability at least $2/3$:
\begin{itemize}
    \item[(i)] $k$-linearity.
    \item[(ii)] $k$-junta.
    \item[(iii)] $k$-sparse  degree-$d$ polynomials.
\end{itemize}
Moreover, for testing $k$-sparse degree-$d$ polynomial, the lower bound is improved to $\Omega\left(\max\{d,k, \nicefrac{1}{\eps}\}\right)$.
\end{restatable}

\subsection*{Discussion}\label{sec:discussion}
In this work, we design efficient algorithms for the problem of testing real-valued functions given via approximate queries, over continuous domains, for three properties: (i) $k$-linearity, (ii) $k$-sparse, low-degree polynomials, and (iii) $k$-juntas. 
Our work opens directions to several interesting questions. 

\begin{itemize}
    \item We note that our results have constraints on the approximate query parameter $\eta$. 
The first open question is whether these can be improved. 

\item Furthermore, our $k$-sparse degree-$d$ polynomial tester performs $\widetilde{O}(d^5 + \nicefrac{d^2}{\eps} + d k^3)$ queries, and our lower bound for this problem is $\Omega\left(\max\{d,k, \nicefrac{1}{\eps}\}\right)$. 
The second open question is whether the gaps in these bounds (w.r.t. $k$, and  $d$) can be improved, e.g., by assuming additional structure on the underlying function, like Lipschitzness, etc. 

\item Another interesting direction is to design tolerant testers~\cite{parnas2006tolerant} for these properties.
This is different from the approximate query testing notion, as in tolerant testing, the decision boundary is expanded to require that functions that are sufficiently close to the property are also accepted with high probability, whereas in the approximate query model, we accept functions $\widetilde{f}$ such that $\dist_{\cD,\infty}(f, \widetilde{f}) \leq \eta$ (pointwise $\eta$-close) which is a stronger constraint compared to the expected $\ell_1$-distance $\dist_{\cD,\ell_1}(f, g)$ that we use between functions.

\item \textcolor{black}{In this work, we have focused on optimizing the query complexity in terms of the sparsity parameter $k$ and degree $d$ (for low-degree polynomial testing). It is an interesting problem to optimize the dependence of $\eta$ in our arguments.}

\item Only our $k$-junta tester has a one-sided error profile, while our $k$-linearity/sparsity testers have two-sided errors. Designing one-sided error testers for these problems is an open problem.

\item Finally, we use the standard Gaussian distribution as the reference distribution. It would be interesting to see if our results can be extended to other concentrated distributions as well.

\end{itemize}

\begin{table}[ht]
    \centering
    \begin{tabular}{|c|c|c|c|}
    \hline
         Problem & Upper Bound 
         & Restriction \\
         \hline
         $k$-linearity & $\widetilde{O}(k \log k + \nicefrac{1}{\varepsilon})$ 
         & $\eta < \min\left\{\eps, O\left(\min_{\substack{i\in[n]:f(\bm e_i)\neq 0}}\frac{|f(\bm e_i)|}{(nk)^2}\right)\right\}$ \\
         $k$-sparsity 
         & $\widetilde{O}(d^5 + \nicefrac{d^2}{\eps} + d k^3)$ 
         & $\eta < \min\{\eps, 1/2^{2^n}\}$ 
         \\
         $k$-junta 
         & $\widetilde{O}\left(\frac{k \log k}{\eps}\right)$ 
         & $\eta < \min\{O(\eps/k^2), O(\nicefrac{1}{k^2 \log^2 k})\}$ \\
         \hline
    \end{tabular}\caption{A comparison of the upper bounds, as well as the corresponding restrictions, for the three sparse representation testing problems. The upper bounds, and the restrictions in the three rows follow sequentially from \autoref{srt_thm:k-lin-intro}, \autoref{srt_thm:k-sparse-intro}, and \autoref{srt_thm:junta-intro}, respectively.
    }
    \label{srt_tab:comparison}
\end{table}
\subsection{New Technical Contributions}\label{sec:new-tech-contri}
\paragraph*{$k$-sparsity Tester:} A critical ingredient in proving the $k$-sparse low degree tester (\autoref{srt_thm:k-sparse-intro}) is a new probabilistic upper bound on the maximum singular value of Hankel matrices associated with sparse, low-degree polynomials (\autoref{srt_defi:Hankel-matrix-poly}). This may be of independent interest,  is sketched out here, and proved in \autoref{sec:sparsity-poly-approx}.

\begin{restatable}[Probabilistic Upper Bound on $\sigma_{\max}$]{theo}{upperboundonsigmamax}\label{lem:upperboundonsigmamax}
Let $f:\R^n\to\R$, $f(\vec{x}) = \sum_{i=1}^{k} a_i M_i(\vec{x})$ be a $k$-sparse, degree-$d$ polynomial, where $M_i$'s are its non-zero monomials, and $\sigma_{\max}(H_t(f,\vec{u}))$ denote the largest singular value of the $t$-dimensional Hankel matrix associated with $f$ at a point $u\in\R^n$, $H_t(f,\vec{u})$. Then, for any $\gamma \in (0,1)$, with $\vec{a} \triangleq (a_1, \cdots, a_k)^\top\in\R^k$,
\[\Pr_{\vec{u} \sim \cN(\bm 0, I_n)}\left[\sigma_{\max}(H_t(f,\vec{u})) \geq \|\vec{a}\|_2^2 \left(2^{d/2} \lceil d/2 \rceil ! + \sqrt{\frac{k}{\gamma}} 2^{d/2} \sqrt{d!}\right)^{2t}\right] \leq \gamma.\]
\end{restatable}

\begin{proof}[Proof Sketch]
For any vector $\bm{v} \in \R^n$, let $|\bm{v}| \in \R_{\geq 0}^n$ denote the vector with the absolute values of the coordinates of $\bm{v}$. Then, for any $\bm{z}, \bm{u} \in \R^n$, using the triangle inequality, and a Vandermonde-like decomposition of $H_{\bm u}\triangleq H_t(|f|, |\bm{u}|)$, (from \autoref{obs:Hankel-charcterization}), where $|f|(\bm{x}) \triangleq \sum_{i \in [k]} |a_i| M_i(\bm{x})$, we can upper-bound $|\bm{z}^\top H_{\bm u} \bm{z}|$ by $|\bm{z}|^\top V^\top D V |\bm{z}| = \|D^{\frac12} V |\bm{z}|\|_2^2$, where $D = \mathrm{diag}(|\bm{a}|)$, $V$ is the Vandermonde matrix $\mathsf{V}_t(|M_1(\bm{u})|,\ldots,|M_k(\bm{u})|) \in \R^{k \times t}$, and $M(|\bm{u}|) = |M(\bm{u})|$ for any monomial $M(\cdot)$. 

By the Courant-Fischer characterization, we have
\[\sigma_{\max} (H_{\bm{u}}) \leq \max_{\|\bm{z}\|_2 \leq 1} \|D^{\frac12} V |\bm{z}|\|_2^2 \leq \max_{\|\bm{z}\|_2 \leq 1} \|D^{\frac12} V\|_{\rm op}^2 \|\bm{z}\|_2^2 \leq \|D^{\frac12} V\|_{\rm F}^2.\]

Using some results from \cite{elandt1961folded} for the moments of the folded normal distribution, along with the properties of the Gamma function and the fact that all the monomials $M_i$ have total degree $\leq d$, we can upper bound 
\[\E_{\bm{u} \sim \cN(\bm{0}, I_n)} \left[|M_i(\bm{u})|\right] \leq 2^{d/2} \lceil d/2 \rceil!, \mbox{ and } \mathrm{Var}[|M_i(\bm{u})|] \leq 2^d \cdot d!.\]

Finally, using Chebyshev's inequality and the union bound (over the $k$ monomials) gives the upper bound for $\sigma_{\max} (H_{\bm{u}})$.
\end{proof}

\begin{rem}\label{sec:prob-upper-bound-extension}
    \autoref{lem:upperboundonsigmamax} seems to be extendable to more general mean-zero distributions of $\bm{u}$ with appropriate concentration (e.g. subgaussian, subexponential) since the crux of the proof, in addition to the use of properties of the Hankel matrix (which do not depend on the distribution of $\bm{u}$), is to bound all the degree $d$ moments of the distribution.
\end{rem}

\paragraph{$k$-linearity Tester:} For $k$-linearity, our tester is similar to the algorithm proposed in \cite{bshouty2023optimal} for testing $k$-linearity of functions $f: \F_2^n \to \F_2$, with the BLR test replaced by the linearity tester from \cite{arora2, arora2024optimaltestinglinearitysosa}. The crucial difficulties in the analysis over the reals come from (i) the $\eta$-approximate query access, and (ii) the use of $\ell_1$-distance for farness. 
We first reject all functions which are far from linearity by means of the linearity test, and, conditioned on the fact that the tester does not reject with high probability, proceed to testing the self-corrected function instead (with closeness to linearity guaranteed). But we now have only approximate query access to the function, which we test for $k$-linearity by splitting the variables into $O(k^2)$ random buckets as in \cite{bshouty2023optimal} and identify the influential ones. 
However, there is no analogous result for real-valued functions. Hence, we present a new analysis in \autoref{claim:close-to-linear-implies-good_new_bucket} for the \textsf{FindInfBucket} (\autoref{alg:findinfbucket}) that performs a binary search over the buckets, but with a different test for influential buckets that accounts for $\eta$-approximate queries. The analysis with approximate queries involves the appropriate use of anti-concentration results (\autoref{srt_thoe:carbery-wright}) for real linear polynomials.

\paragraph{$k$-junta tester:} Our tester is inspired by the junta testers used in \cite{blais2009testing} and \cite{blais2015partially}, but our definition of influence (\autoref{defn:influence}) is different ($\ell_1$ rather than Hamming), to work with approximate query oracle. \autoref{thm:structure-junta}, \autoref{lem:infl-submodular}, and \autoref{lem:influence_partition_complement}, in our analysis are similar to that of \cite{blais2015partially}. The significant difference in analysis for the $\ell_1$ distance comes in our \autoref{claim:close-to-linear-implies-good_new_bucket_junta}, and the subsequent proof of \autoref{theo:kjuntafinalmain}.

\subsection*{Organization of the paper}
In \Cref{sec:overview}, we present an overview of our results and techniques, followed by a discussion of the preliminaries required, in \Cref{sec:prelim}. 
In \Cref{sec:klinearitytest}, we present our $k$-linearity tester, which is followed by our tester 
for $k$-sparse low-degree polynomials in \Cref{sec:ksparsepolynomialtest}, and our $k$-junta tester in \Cref{sec:kjuntatest}. 
Finally, in \Cref{sec:lowerbound}, we present the lower bounds. 
Some associated proofs and subroutines from prior works are moved to the appendix for brevity.

\section{Technical Overview}\label{sec:overview}

We present a brief overview of our techniques, starting with $k$-linearity testing.
\paragraph{Testing $k$-linearity}
We build upon the \emph{self-correct and test} approach of \cite{blum1990self,halevy2007distribution}. 
Instead of directly testing if $f$ is $k$-linear, we construct a function $g_{\text{self-correct}}$ such that if $f$ is $k$-linear, then $g_{\text{self-correct}}$ will also be $k$-linear. 
Moreover, we simulate queries to $g_{\text{self-correct}}$
using queries to $f$, and test this newly constructed function $g_{\text{self-correct}}$.

We use the Gaussian distribution as the reference distribution, since there are no uniform distributions over continuous domains with infinite support. 
(This deviates from the self-correction approach of \cite{blum1990self}.)
In particular, we evaluate $f$ on a set of points sampled from $\cN(\bm 0,I_n)$ to construct the self-corrected function $g_{\text{self-correct}}$. 
To deal with the fact that different points from $\cN(\bm 0,I_n)$ have different probability masses, the idea is to radially project the sampled points from $\cN(\bm 0,I_n)$ into a small Euclidean ball $\ball(\bm 0,\srad)$ of a small (constant) radius ($r=1/50$ suffices) such that the probability masses of all points sampled from that ball is roughly the same. 
Moreover, since we work with an approximate oracle, we use the following notion of the self-corrected function,
\[
g(\bm p) \triangleq \
\kappa_{\bm p} \cdot \mathop{\mathsf{med}}_{\bm x\sim\mathcal{N}(\bm 0,I_n)} \left[  f \left( \frac{\bm p}{\kappa_{\bm p}}- \bm x \right) +f \left (\bm x \right)  \right]
\]
where $\kappa_{\bm p}:\R^n \to \R$ is a contraction factor, defined as:
\[\kappa_{\bm p}\triangleq\begin{cases}
    1 &\text{, if }\|\bm p\|_2\leq r \\
    \lceil \|\bm p\|_2/r\rceil &\text{, if }\|\bm p\|_2 > r
\end{cases}    \]
so that $\bm p/\kappa_{\bm p}\in \ball(\bm 0,r)$, and \textcolor{black}{$\mathsf{med}$ denotes the median function}. This definition of self-correction function was used in \cite{fleming2020distribution, arora2, arora2024optimaltestinglinearitysosa}.

To test $k$-linearity (\autoref{srt_balg:testklinear}), we first test if $f$ is pointwise close to some additive (aka. linear) function using \textsc{Approximate Additivity Tester} (\autoref{alg:zero-mean-additivity}). 
If it rejects $f$, we also reject $f$. 
However, if \autoref{alg:zero-mean-additivity} does not reject, then the self-corrected function $g$ is pointwise close to some linear function. 
As we only have access to an approximate oracle access to $f$, we can't simulate $g$ exactly.  
So, we use \textsc{Approximate-$g$}, the approximate query oracle for $g$ (part of \autoref{alg:subroutines_noisy}). 
The work of \cite{arora2, arora2024optimaltestinglinearitysosa} proves: (i) $g$ and \textsc{Approximate-$g$} are pointwise close in $\ball(\bm 0,r)$, and (ii) $f$ and \textsc{Approximate-$g$} are also pointwise close. 
So, $f$ is pointwise close to \textsc{Approximate-$g$}. 

We partition the $n$-variables $[n]$ into $k^2$ buckets uniformly at random. As a result of the partition, if $f$ is $k$-linear, the influential variables will be separated into different buckets w.h.p. We can then detect them via the subroutine $\findbucket$ (\autoref{alg:findinfbucket}), which recursively isolates the variables in a bucket and returns only the influential ones.  We then reject if the number of buckets with influential variables (variables whose values determine the value of $f$) is more than $k$.

\paragraph{Testing $k$-sparsity}
Our $k$-sparse, degree-$d$ polynomial tester (\autoref{srt_alg:testksparse}) adopts a similar approach. 
We first test if $f$ is a low-degree polynomial, using the \apxlowdegtest (\autoref{alg:low_degree_main_algorithm_approx}) from \cite{arora2, arora2024optimaltestinglinearitysosa}. 
If it rejects $f$, we also reject $f$. 
However, if \autoref{alg:low_degree_main_algorithm_approx} does not reject $f$, then $f$ is point-wise close to a low-degree polynomial.

As in $k$-linearity testing, we use a self-corrected function $g$ from \cite{arora2, arora2024optimaltestinglinearitysosa}: 
For points $\bm p \in \ball(\bm 0, \srad)$, 
$ g(\bm p)$ is the (weighted) median value of 
$g_{\bm q}(\bm p) \triangleq \sum_{i=1}^{d+1}(-1)^{i+1}\binom{d+1}{i} f(\bm p+i\bm q)$, 
weighted according to the probability of $\bm q \sim \mathcal{N}(\bm 0,I_n)$, i.e., 
\[ g(\bm p) \triangleq \underset{\bm q \sim \mathcal{N}(\bm 0, I_n)}{\mathsf{med}}[g_{\bm q}(\bm p)].\]
Intuitively, $g_{\bm q}(\bm p)$ is the value that $f$ should take if, when restricted to the line $\pline_{\bm p,\bm q}\triangleq\{\bm p+t\bm q,t\in\R\}$, $f$ would be a degree-$d$ univariate polynomial. 
Taking the weighted median over all directions $\bm q\sim\cN(\bm 0,I_n)$, ensures that the self-correction proportionately respects the values of $f$, in a local neighborhood of $\bm p$. 
For $\bm p\not\in \ball(\bm 0,r), g$ is defined via radial extrapolation from within $\ball(\bm 0, \srad)$ along the radial line $\pline_{\bm 0, \bm p}$.

Given exact query access to $f$, we can simulate query access to the self-corrected function $g$. 
As we only have approximate query access to $f$, we use $\apxqueryg$ (\autoref{alg:subroutines_approx}, the approximate oracle to the self-corrected function associated with $f$), which was proved to be pointwise close to $g$. 
As a result, $f$ will be pointwise close to $\apxqueryg$ as well. 

Once we have that $f$ is close to a low-degree polynomial, we use a Hankel matrix (\autoref{srt_defi:Hankel-matrix-poly}) based characterization for sparse polynomials. 
\cite{gjr10,bot88} proved \autoref{obs:Hankel-charcterization}: a polynomial $f$ (over finite fields) is $k$-sparse, if and only if its associated Hankel matrix has a non-zero determinant. 
This can be efficiently tested with only $2k+1$ queries to $f$.
For polynomials over the reals with exact-queries, the idea from \cite{gjr10} can be extended to give a zero-error (perfect soundness and completeness), zero-gap (does not require $\varepsilon$-farness for $\varepsilon$ > 0 in the No case) tester, making $2k + 1$ queries for deciding whether (i) $f$ is $k$-sparse, or (ii) $f$ is not $k$-sparse (proved in \Cref{srt_lem:perfect-sparsity-tester}). This result is even stronger than what \cite{gjr10} gets for finite fields, since the zero set of any polynomial over the reals has zero Gaussian (or Lebesgue) measure, whereas the corresponding probability needs to be bounded in terms of the degree of the polynomials over finite fields (e.g. by the Schwartz-Zippel lemma).
Unfortunately, since we only have approximate query access to $f$, this technique no longer works, as determinants of sums of matrices do not behave nicely. 
So we take a probabilistic approach, showing that the \emph{noisy} Hankel matrix constructed from the approximate query to $f$ is not too far from the exact Hankel matrix (proved in \autoref{obs:approx-hankel-structure}). 
Finally, we show in \autoref{lem:approx-query-sparsity-testing-poly} using Weyl's inequality (\autoref{srt_thm:weyls-ineq}), that if $f$ is a $k$-sparse, low-degree polynomial, then the smallest eigenvalue of the \emph{noisy} Hankel matrix associated with $\apxqueryg$ is not too large.
Combining them all, our main result is proved in \autoref{sec:k-sparsity-tester-analysis}.

Since, we only have access to the Hankel matrix of $\widetilde{f}=\apxqueryg$ which is $\eta^\prime$-close to $f$ (on all the query points, whp),
we can express $\widetilde{H}_{\bm{u}} \triangleq H_{t}(\widetilde{f}, \bm{u})$ (for $t = k+1$) of such a $\widetilde{f}$ as the sum of $H_{\bm{u}} \triangleq H_{t}(f, \bm{u})$ and a Hankel-structured error matrix with bounded spectral/operator norm (\autoref{obs:approx-hankel-structure}). From this, and Weyl's inequality, we get completeness as long as we \textsc{Accept} when the smallest singular value $\sigma_{\min} (\widetilde{H}_{\bm{u}}) \leq \eta^\prime(k + 1)$ in all $\Theta(d(k+1)^2)$ rounds.

To prove soundness of the tester, we need to show that $\sigma_{\min} (\widetilde{H}_{\bm{u}}) > \eta^\prime (k+1)$ in some round with probability $\geq 2/3$, as long as $Q(\bm{u}) \triangleq \det(H_{\mathbf{u}})$ is a non-zero polynomial. This analysis (in \autoref{lem:approx-query-sparsity-testing-poly}) follows by invoking the probabilistic upper bound on $\sigma_{\max} (H_\mathbf{u})$, for $\bm{u} \sim \cN(\bm 0, I_n)$.

\paragraph{Testing $k$-junta}
Finally, we discuss our algorithm for testing $k$-juntas (\autoref{alg:testjuntaapprox}). 
Our approach is to first randomly partition the $n$-variables into $k^2$ buckets. 
If $f$ is a $k$-junta, the $k$ influential variables will be separated into distinct buckets w.h.p (by the birthday paradox).

Now we run $O(k/\eps)$ iterations to find if there exists any influential variable in any bucket, using the subroutine $\findbucket$ (\autoref{alg:findinfbucket}). 
After these iterations, if we find more than $k$ influential variables in $f$, we reject it.
Otherwise, we accept $f$.
Our analysis follows a combinatorial style similar to \cite{blais2015partially}. 
We would like to note that although we use $\findbucket$ for $k$-linearity testing as well, the analysis here significantly deviates from that of $k$-linearity testing and is presented in \autoref{claim:close-to-linear-implies-good_new_bucket_junta}. The main result is formally proved in \autoref{sec:k-junta-analysis}.

\section{Preliminaries}\label{sec:prelim}

\paragraph{Notations} Throughout this work, we use boldface letters to represent vectors of length $n$ and normal face letters for variables. 
Specifically, $\bm e_i\triangleq(0,\hdots,0,1_i,0,\hdots,0)$ denotes the $i^{th}$ standard unit vector. 
For $n\in\N$, let $[n]$ denote the set $\{1, \dots, n\}$. 
For a matrix $A$, let $\|A\|_{\infty}$, $\|A\|_{\mathrm{op}}$, and $\|A\|_{\mathrm F}$ denote the supremum, operator, and Frobenious norms of $A$, respectively. 
See \cite{horn2012matrix} for the formal definitions. 
For concise expressions and readability, we use the asymptotic complexity notion of $\widetilde{O}(\cdot)$, where we hide poly-logarithmic dependencies of the parameters.
For any $f:\R^n\to\R$, let $\|f\|_{-\infty(/\infty),C}$ denote the infimum(resp. supremum) value of $f$ over some $C\subseteq\R^n$.
Let $\Pi_n$ be the class of functions $\R^n \rightarrow \R$ satisfying some particular first-order property and let $\Pi = \cup_{n \geq 1} \Pi_n$. 
\begin{defi}[$\ell_1$-distance]\label{srt_def:l1-distance}
Let $\caD = \{\caD_n\}_{n \geq 1}$ is family of distributions with $\caD_n$ being a distribution on $\R^n$. For two arbitrary functions $f, g: \R^n \rightarrow \R$, the $\ell_1$-distance between $f$ and $g$ is defined as:
\[\mathsf{dist}_{\caD,\ell_1}(f, g) \triangleq \E_{\bm x \sim \caD_n}\,[|f(\bm x) - g(\bm x)|].\]
\end{defi}

We also define the $\ell_p$-distance of $f$ to the class $\Pi_n$, and hence the class $\Pi$, by
\[\mathsf{dist}_{\caD,\ell_p}(f,\Pi)\triangleq\mathsf{dist}_{\caD,\ell_p}(f,\Pi_n) \triangleq \inf_{g \in \Pi_n} \mathsf{dist}_{\caD,\ell_p}(f, g).\]

\begin{rem}
    Since our distributions are supported over continuous spaces, the distances we consider are also defined over such continuous spaces, i.e., $\dist_{\caD,\ell_j}(\cdot,\cdot)\equiv\dist_{\caD, L_j}(\cdot,\cdot)$, for all $j$.
\end{rem}

We have $\mathsf{dist}_{\caD,\ell_p}(f, \Pi) = 0$ if and only if there exists a function $g \in \Pi_n$ which agrees with $f$ almost everywhere with respect to $\caD_n$ (in the measure-theoretic sense). We will only concern ourselves with $p\in\{0,1\}$ since we are dealing with scalar-valued functions (aka functionals) and hence $\|f(\bm x) - g(\bm x)\|_p = |f(\bm x) - g(\bm x)|$ for all $p > 0$.

The notion of $\ell_1$-distance makes more sense in the case of real-valued functions, especially if the evaluation of $f$ is not exact. To see this, consider the case when $g = f + \varepsilon$ for some $\varepsilon > 0$, then $\mathsf{dist}_{\caD,\ell_1}(f,g) = \varepsilon$ whereas $\mathsf{dist}_{\caD,\ell_0}(f,g) = 1$.

We have the following observation that shows that over finite fields, any $(k+1)$-linear function is always far from any $k$-linear function in $\ell_0$-distance.

\begin{restatable}{obs}{distlzero}\label{srt_obs:distlzero}
For $1 \leq k < n$, for any $(k+1)$-linear function $f: \F_2^n \rightarrow \F_2$, we have $\mathsf{dist}_{\ell_0}(f, g) = 1/2$ for any $k$-linear function $g: \F_2^n \rightarrow \F_2$.
\end{restatable}

\begin{proof}
Let $f(\vec{x}) = \sum_{i \in [n]} c_i x_i$, and $S \triangleq \{i \in [n] : c_i \neq 0\}$ with $|S| = k+1$, since $f$ is $(k+1)$-linear. For any $k$-linear $g$ with $g(\vec{x}) = \sum_{i \in [n]} a_i x_i$, there is at least one $i^\ast \in S$ for which $a_{i^\ast} = 0$. Then
\begin{align*}
\mathsf{dist}_{\ell_0}(f,g) &= \Pr_{\bm x \leftarrow F_2^n}\left[f(\bm x) \neq g(\bm x)\right] = \Pr_{\bm x \leftarrow F_2^n}\left[f(\bm x) + g(\bm x) = 1\right] \tag{Note $c_{i^\ast}=1,a_{i^\ast}=0$}\\
&=\Pr_{\bm x \leftarrow F_2^n}\left[ \sum_{i=1}^{n} (a_i+c_i)x_i = 1\right] = \Pr_{\bm x \leftarrow \F_2^n}\left[x_{i^\ast} + \sum_{i \neq i^\ast} (a_i+c_i) x_i = 0\right] = \frac{1}{2}.\qedhere 
\end{align*}
\end{proof}

Unfortunately, there is no general analogous result in the real case with the $\ell_1$-distance, but there is an upper bound, provided the distribution $\cD_n$ has some concentration.

\begin{restatable}{claim}{distlone}
If $f, g: \R^n \rightarrow \R$ are linear functions with $f(\vec 0) = g(\vec 0)$, then
\[\|f - g\|_{-\infty,\textrm{supp}(\cD)} \leq \mathsf{dist}_{\cD,\ell_1}(f, g) \leq \|f - g\|_2 \cdot \E_{\vec{x} \sim \cD_n} \|\vec{x}\|_2. \]
\end{restatable}

\begin{proof}
Suppose $f(\vec{x}) = \sum_{i \in [n]} a_i x_i$ and $g(\vec{x}) = \sum_{i \in [n]} b_i x_i$ for some $\vec{a}, \vec{b} \in \R^n$. Let $h = f - g$, so that $h(\vec{x}) = \sum_{i \in [n]} (a_i - b_i) x_i = \iprod{\vec{a} - \vec{b}}{\vec{x}}$. We have, by Cauchy-Schwarz inequality,
\[\mathsf{dist}_{\cD,\ell_1}(f, g) \triangleq \E_{\vec{x} \sim \cD_n}\,\left|h(\vec{x})\right| \leq \|\vec{a} - \vec{b}\|_2 \cdot \E_{\vec{x} \sim \cD_n} \|\vec{x}\|_2.\]
and the fact that $\E[|h(\vec{x})|] \geq \inf_{\vec{x}\sim\cD} |h(\vec{x})|$ gives the lower bound.
\end{proof}

Our definitions hold for the general reference distributions $\caD$ which are suitably concentrated. Later, we work with the standard Gaussian distribution $\cN(\bm 0,I_n)$, which also has this desired concentration property.

\begin{defi}[Concentrated distribution]\label{defi:distribution-concentration}
Let $\eps\in(0,1),R\geq 0$, and $\bm c\in\R^n$.
A distribution $\caD$ supported on $\R^n$ is $(\eps,R,\bm c)$-\emph{concentrated} if most of its mass is contained in a ball of radius $R$ centered at some point $\bm c\in\R^n$, i.e.,
\[\Pr_{\bm p\sim\caD}[\bm p\in \ball(\bm c,R)]\geq 1-\eps.\]
\end{defi}

\begin{fact}[{\cite[Theorem 2.9]{blum_hopcroft_kannan_2020}}]\label{fact:gaussian-concentration}
    The standard Gaussian distribution $\cN(\bm 0,I_n)$ is $(0.01,2\sqrt{n},\bm 0)$-concentrated.
\end{fact}

For brevity, we may omit the third entry corresponding to the center of the ball, when the center is the origin. For example,  $(0.01,2\sqrt{n})$-concentrated means $(0.01,2\sqrt{n},\bm 0)$-concentrated.

\begin{defi}[Property Tester]\label{srt_def:low_degree_sparsity_tester}
Let $\caP$ be a real function property. 
An algorithm is said to be a \emph{tester} for $\caP$ with respect to distance measure $\dist(\cdot,\cdot)$; with proximity parameter $\eps > 0$, completeness error $c \in (0, 1)$, and soundness error $s \in (0, 1)$ if, given query access (either exact, or $\eta$-approximate query access) to a function $f:\R^n \to \R$, and sample access to a reference distribution $\caD_n$ on $\R^n$, the algorithm performs $q(n,d,k,\eps,c,s)$ queries to $f$ and:
\begin{enumerate}
	\item[(i)] Outputs \textsc{Accept} with probability $\geq 1 - c$ (over the randomness of the algorithm), if $f \in \caP$.
 
	\item[(ii)] Outputs \textsc{Reject} with probability $\geq 1 - s$, if $\dist(f, g) \geq \varepsilon$ for all $g \in \caP$.
\end{enumerate}
\end{defi}

\subsection{Preliminaries on Polynomials and Linear Algebra}

We briefly discuss some notions, and properties of polynomials:

\begin{defi}[Monomials]
Suppose $x_1,\ldots,x_n$ are indeterminates. A \emph{monomial} in these indeterminates is a product of the form $x_1^{\alpha_1} x_2^{\alpha_2} \cdots x_n^{\alpha_n}$ where $\bm \alpha \triangleq (\alpha_1, \ldots, \alpha_n) \in \N^n$, which we also denote by $\bm x^{\bm \alpha}$. The total degree of monomial $\bm x^{\bm \alpha}$ is $\|\bm \alpha\|_1 \triangleq \alpha_1 + \cdots + \alpha_n$. The degree of $\bm x^{\bm \alpha}$ in the variable $x_i$ is $\alpha_i$. The individual degree of $\bm x^{\bm \alpha}$ is $\|\bm \alpha\|_\infty \triangleq \max_{i \in [n]} \alpha_i$. Over a field $\mathbb{F}$, each monomial $\bm x^{\bm M} = x_1^{M_1} \cdots x_n^{M_n}$ for $\bm M \in \N^n$ corresponds to a function $\bm M: \F^n \rightarrow \F$, $\bm M(\vec{z}) \triangleq z_1^{M_1} \cdots z_n^{M_n}$.
\end{defi}

\begin{defi}[Polynomials]
A \emph{polynomial} in $x_1,\ldots,x_n$ (aka an $n$-variate polynomial) over a field $\F$ is a finite $\F$-linear combination of monomials in $x_1,\ldots,x_n$. That is, an $n$-variate polynomial over $\F$ is of the form $P(\bm x) = \sum_{i=1}^{m} a_i\bm x^{\bm M_i}$ where $m \geq 0$, $\bm M_1,\ldots,\bm M_m \in \N^n$ and $a_1,\ldots,a_m \in \F \setminus \{0\}$. The set of all such polynomials is denoted as $\F[x_1,\ldots,x_n]$.
\end{defi}

\begin{defi}[Polynomial sparsity]\label{srt_defi:poly_sparsity}
The \emph{sparsity} of a polynomial $P(\bm x) = \sum_{i=1}^{m} a_i\bm x^{\bm M_i}$, denoted $\|P\|_0$, is the number of non-zero coefficients $a_i$ in its monomial-basis representation. The unique polynomial $P$ with $\|P\|_0 = 0$ is the \emph{zero polynomial}, denoted by $P \equiv 0$.
\end{defi}

\begin{defi}[Total and individual degrees]\label{srt_defi:poly_degree}
The \emph{total degree} $\deg(P)$ of a non-zero polynomial $P(\bm x) = \sum_{i \in [m]} a_i\bm x^{\bm M_i}$ is the maximum total degree of its monomials, i.e., $\deg(P) \triangleq \max_{i \in [m]} \|\bm M_{i}\|_1$. Similarly, the \emph{individual degree} $\ideg(P)$ of $P$ is the maximum individual degree of its monomials, i.e., $\ideg(P) \triangleq \max_{i \in [m]} \|\bm M_{i}\|_\infty$. The total, as well as the individual degree of the zero polynomial $P \equiv 0$ is defined to be $-\infty$.
\end{defi}

In this work, we will be primarily working on the total degree. So, we will use degree to represent total degree when it is clear from the context.

When $\F$ is an infinite field, each polynomial $P(x_1,\hdots,x_n) = \sum_{i \in [m]} a_i\bm x^{\bm M_i}$ over $\F$ corresponds to a \emph{unique} $\F$-valued function $P: \F^n \rightarrow \F$ with $P(\vec{z}) = \sum_{i=1}^{m} a_i\bm M_i(\vec{z}) = \sum_{i=1}^{m} a_i \vec{z}_1^{M_{i,1}} \cdots \vec{z}_n^{M_{i,n}}$. Such functions are referred to as \emph{polynomial functions}. Unlike with finite fields, the equality of two real polynomial functions implies the equality of the formal polynomials. Thus, we use the formal polynomial $P \in \R[x_1,\ldots,x_n]$ and the 	\emph{polynomial function} $P: \R^n \rightarrow \R$ interchangeably.

\begin{defi}\label{srt_defi:deg_d_poly_class}
Let $\cP^{\mathrm{tot}}_{n,d,k}$ denote the class of all $n$-variate real polynomials with \emph{total degree} $\leq d$ and sparsity $\leq k$.
\[\cP^{\mathrm{tot}}_{n,d,k} \triangleq \left\{P \in \R[x_1,\ldots,x_n] \mid P = \sum_{i=1}^{k} a_i\bm x^{\bm M_i} : \forall\, i \in [k], a_i \in \R, \bm M_i \in \N^n: \|\bm M_{i}\|_1 \leq d\right\}.\]
Similarly, let $\cP^{\mathrm{ind}}_{n,d,k}$ denote the class of all $n$-variate real polynomials with \emph{individual degree} $\leq d$ in all variables and sparsity $\leq k$.
\[\cP^{\mathrm{ind}}_{n,d,k} \triangleq \left\{P \in \R[x_1,\ldots,x_n] \mid P = \sum_{i=1}^{k} a_i\bm x^{\bm M_i} : \forall\, i \in [k], a_i \in \R,\bm M_i \in \{\{0\} \cup [d]\}^{\otimes n}\right\}.\] 
\end{defi}

We will use the notion of Vandermonde matrices in this work.

\begin{defi}[Vandermonde matrix and determinant]\label{srt_defi:vandermonde}
For any $m, n > 0$ and $x_1,\ldots,x_m \in (R, +, \cdot)$, the $m \times n$ Vandermonde matrix with nodes $x_1,\ldots,x_m$ denoted by $\mathsf{V}_n(x_1,\ldots,x_m)$ is:
	\[\mathsf{V}_n(x_1,\ldots,x_m) \triangleq 
    \begin{pmatrix}
    1 & x_1 & \cdots & x_1^{n-1}\\
    1 & x_2 & \cdots & x_2^{n-1}\\ 
    \vdots & \vdots & \ddots & \vdots\\
    1 & x_m & \cdots & x_m^{n-1}
    \end{pmatrix}.\]
	If $m=n$, $\det\paren{\mathsf{V}_n(x_1,\ldots,x_n)} = \prod_{1 \leq i < j \leq n} (x_j - x_i)$, which is the Vandermonde determinant formula.
\end{defi}

\subsection{Results from Perturbation Theory, and Measure Theory}

For analyzing our testers, we use the following result from perturbation theory.

\begin{theo}[Weyl's Inequality~\cite{horn1994topics}]\label{srt_thm:weyls-ineq}
    Let $\hat{\Sigma} = \Sigma + E$, where $\Sigma, E\in\R^{n \times n}$ are both symmetric matrices. Let $\lambda_i(A)$ denote the $i^{th}$ eigenvalue of (symmetric matrix) $A$, sorted in non-increasing order, and let $\|A\|_{\mathrm{op}}$ denote its $\ell_2$-operator (spectral) norm. Then,
    \[
    \max_{i\in [n]} \{|\lambda_i(\hat{\Sigma}) - \lambda_i(\Sigma)|\} \leq \|E\|_{\mathrm{op}} \leq \max\{|\lambda_d(E)|,|\lambda_1(E)|\}.
    \]
\end{theo}

The following celebrated inequality of Carbery and Wright~\cite{CarberyWright2001Distributional} provides anti-concentration bounds for polynomials of i.i.d Gaussian random variables (more generally, log-concave), which will be used in our proofs, particularly in the analysis of our $k$-linearity tester.

\begin{theo}[{\cite[Theorem 8]{CarberyWright2001Distributional}}]\label{srt_thoe:carbery-wright}
    Let $f:\R^n\to\R$ be a polynomial of degree $d$, such that $\var_{\cN(\bm 0,I_n)}[f]=1$. 
    Then, for any $t\in\R$, and any $\eps>0$,
    \[\Pr_{\bm x\sim\cN(\bm 0,I_n)}[|f(\bm x)-t|\leq\eps]\leq O(d)\eps^{1/d}.\]
\end{theo}

A recent result of Glazer and Mikulincer~\cite{glazer2022anticonc} provided a more suitable form for applying Carbery-Wright to polynomials, where the variance is lower bounded in terms of the coefficients.
\begin{theo}[{\cite[Corollary 4]{glazer2022anticonc}}]\label{srt_thm:glazer-mikulincer}
If $f: \R^n \rightarrow \R$ is a polynomial $f(\bm x) = \sum_{i=1}^{k} a_i \bm x^{\bm M_i}$ of degree $d$, then there exists an absolute constant $C > 0$ such that for any $t \in \R$ and $\eps > 0$,
\[
\Pr_{\vec{x} \sim \cN(\bm 0,I_n)}[|f(\vec{x}) - t| \leq \eps] \leq Cd\left(\frac{\eps}{\mathrm{coeff}_d(f)}\right)^{1/d}\text{, where }
\mathrm{coeff}^2_d(f) \triangleq \sum_{\substack{i \in [k]: \|\bm M_i\|_1 = d}} a_{i}^2.
\]
\end{theo}

The Paley-Zygmund anti-concentration result will be useful for analyzing our junta tester. 
\begin{theo}[\cite{Paley_Zygmund_1932}]\label{thm:paley-zygmund}
    For a random variable $X\geq 0$ with finite variance, for all $\theta\in[0,1]$,
    \[\Pr[X\geq\theta\E[X]]\geq(1-\theta)^2 \frac{\E^2[X]}{\E[X^2]}.\]
\end{theo}

\section{\texorpdfstring{$k$}{k}-linearity Testing}\label{sec:klinearitytest}

In this section, we present and analyze our algorithm \testlin (\autoref{srt_balg:testklinear}).

\begin{theo}[Generalization of \autoref{srt_thm:k-lin-intro}]\label{srt_theo:k_linearity}
Let $k \in \N$, \textcolor{black}{$f : \R^n \rightarrow \R$ is a function bounded in the
ball $B(\bm 0, 2 \sqrt{n})$}, and $\eps, \eta \in (0,2/3)$ be such that $\eta < \min\left\{\eps, O\left(\min_{\substack{i\in[n]:f(\bm e_i)\neq 0}}\frac{|f(\bm e_i)|}{(nk)^2}\right)\right\}$, where $\bm e_i$ denotes the $i^{th}$ standard unit vector. Given $\eta$-approximate query access to $f$, there exists a tester \testlin (\autoref{srt_balg:testklinear}) that performs $\widetilde{O}(k \log k + \nicefrac{1}{\varepsilon})$ queries and guarantees:
\begin{itemize}
    \item \textsc{Completeness:} If $f$ is a $k$-linear function, then \autoref{srt_balg:testklinear} \textsc{Accepts} with probability at least $2/3$.

    \item \textsc{Soundness:} If $f$ is $\eps$-far from all $k$-linear functions, then \autoref{srt_balg:testklinear} \textsc{Rejects} with probability at least $2/3$.

\end{itemize}

\end{theo}

\begin{algorithm}[h]
\caption{\testlin}\label{srt_balg:testklinear}
{\bf Inputs:} $\eta$-approximate query oracle $\widetilde{f}$ for $f: \R^n \rightarrow \R, k\in\N$, proximity parameter $\eps \in (0,1)$.\

{\bf Output:} Returns \textsc{Accept} iff $f$ is a $k$-linear function.

Run \textsc{Approximate Additivity Tester}($f,\cN(\bm 0,I_n),\eta,\eps,2\sqrt{n}$) (\autoref{alg:zero-mean-additivity}). \

If \autoref{alg:zero-mean-additivity} rejects, \Return \textsc{Reject}. \

Set $r= O(k^2)$. \

Bucket each $i\in[n]$ into $B = \{B_1,\hdots,B_r\}$, uniformly at random.\Comment{$[n]=\uplus_{i=1}^r B_i$.}\

$\mathsf{InfBuckets}=\mathsf{FindInfBuckets}(\textsc{Approximate-$g$},B,[r])$.\Comment{$\findbuckets=$ \autoref{alg:findinfbuckets}}\

\If{$|\mathsf{InfBuckets}| > k$}{

\Return \textsc{Reject}.

}
\Return \textsc{Accept}. 
                           
\end{algorithm}

\autoref{srt_balg:testklinear} uses two subroutines $\findbucket$ (\autoref{alg:findinfbucket}) and $\findbuckets$ (\autoref{alg:findinfbuckets}).
We will first present and analyze them in \autoref{sec:analyses-subroutines}, and then analyze \autoref{srt_balg:testklinear} in \autoref{sec:k-lin-tester-analysis}.

A closely related (and \emph{efficiently testable}) notion of \emph{additivity} may be noted here.
\begin{defi}[Additive function]
A function $f:A\to B$ is \emph{additive}, if for all $x,y\in A, f(x\oplus_A y)=f(x)\oplus_B f(y)$, where $\oplus_A$ and $\oplus_B$ denote the bitwise-xor operations in $A$ and $B$, respectively.
\end{defi}

Over finite domains, additivity implies linearity. But over continuous domains, this isn't always the case. However, for continuous functions, testing additivity suffices from the following fact:

\begin{fact}[{\cite[Section 5.2]{kuczma2009inequalities}}]
For continuous functions, additivity is equivalent to linearity.
\end{fact}

\autoref{srt_balg:testklinear} uses a query oracle to the self-corrected function, \textsc{Approximate-$g$} (presented in \autoref{alg:subroutines_approx}). We present a brief discussion of self-correction for additivity testing:

\paragraph*{Self-Correction} We use the definition of self-correction from \cite{arora2, arora2024optimaltestinglinearitysosa}, who studied the problem of testing additive (linear) functions, given approximate oracle access. 
Let $r$ be a sufficiently small rational; $r \triangleq 1/50$ suffices. 
Define the value of the self-corrected function $g$ at a point $\bm{p}\in \ball(\bm 0,r)$ as the (weighted) median value of $g_{\bm{x}}(\bm{p})\triangleq f(\bm{p}-\bm{x}) + f(\bm{x})$, each weighted according to its probability mass under $\bm{x} \sim \mathcal{N}(\bm 0,I_n)$. 
For points $\bm p\not\in \ball(\bm 0,r)$, we project them into the ball by scaling by a sufficiently large contraction factor that depends on the magnitude of $\bm p$. 
\begin{equation}\label{eqn:gselfcorrection}
g(\bm p) \triangleq \kappa_{\bm p} \cdot \mathop{\mathsf{med}}_{\bm x\sim\mathcal{N}(\bm 0,I_n)} \left[  g_{\bm x} \left( \frac{\bm p}{\kappa_{\bm p}} \right)  \right] = \kappa_{\bm p} \cdot \mathop{\mathsf{med}}_{\bm x\sim\mathcal{N}(\bm 0,I_n)} \left[  f \left( \frac{\bm p}{\kappa_{\bm p}}- \bm x \right) +f \left (\bm x \right)  \right],   
\end{equation}

where $\kappa_{\bm p}:\R^n \to \R$ is defined as $\kappa_{\bm p}\triangleq\begin{cases}
    1 &\text{, if }\|\bm p\|_2\leq r \\
    \lceil \|\bm p\|_2/r\rceil &\text{, if }\|\bm p\|_2 > r
\end{cases}$, so that $\bm p/\kappa_{\bm p}\in \ball(\bm 0,r)$.

We note the following results from \cite{arora2, arora2024optimaltestinglinearitysosa} about their approximate additivity tester.

\begin{theo}[{\cite[Theorem D.1]{arora2}}, and {\cite[Theorem 3.3]{arora2024optimaltestinglinearitysosa}}]\label{srt_thm:additivity-tester-arora2}
Let $\alpha, \varepsilon > 0$, \textcolor{black}{for $L > 0$, suppose $f : \R^n \rightarrow \R$ is a function that is bounded in the
ball $B(0, L)$} and $\cD$ be an unknown $(\varepsilon/4,R)$-concentrated distribution. There exists a one-sided error, $O (\nicefrac{1}{\varepsilon})$-query tester (\autoref{alg:zero-mean-additivity}) which with probability at least $99/100$, distinguishes when $f$ is pointwise $\alpha$-close to some additive function and when, for every additive function $h$, $\Pr_{\bm p \sim \cD}[|f(\bm p)-h(\bm p)| > O(R n^{1.5}\alpha)] > \eps$.
\end{theo}

\begin{lem}[{\cite[Lemma D.3 and D.6]{arora2}}]\label{srt_lem:g-is-additive-arora2}
If \textsc{TestAdditivity}{$(f,3\alpha)$} (\autoref{alg:subroutines_noisy}) accepts with probability at least $1/3$, then $g$ is a $42\alpha$-additive function inside the small ball $\ball(\bm 0,r)$, and furthermore, for every $\bm p\in\ball(\bm 0,r)$  it holds that 
  \[\Pr_{\bm x \sim \mathcal{N}(\bm 0,I_n)}[\left|g(\bm p) - g_{\bm x}(\bm p)\right|\geq 12\alpha ] < 12/125.\]
\end{lem}

\begin{lem}[{\cite[Lemma D.4 and D.5]{arora2}}]\label{srt_lem:g-additive-in-ball}
  If \textsc{TestAdditivity}($f, 3\alpha$) accepts with probability at least $1/3$, then for every $\bm p,\bm q \in \ball(\bm 0,r)$ with $\|\bm p + \bm q\|_2 \leq \srad$, it holds that \[\left|g(\bm p + \bm q) -g(\bm p) - g(\bm q)\right|\leq 42\alpha.\]
\end{lem}
The following lemma gives us a way to scale the closeness for degree-$d$ polynomials.

\begin{lem}[{\cite[Lemma 4.19]{arora2}}]\label{lem:extrapolation_approx}
    Let $R>r'>0$ be any real numbers. If $g$ is pointwise $\eta$-close to a degree-$d$ polynomial in $\ball( \bm 0, r')$, then $g$ is pointwise $(12R/r')^d\eta$-close to a degree-$d$ polynomial on all points in $\ball(\bm 0,R)$.
\end{lem}

Using \Cref{srt_lem:g-is-additive-arora2}, \autoref{lem:extrapolation_approx} and the structure of \textsc{Approximate-$g$} subroutine, they bound the distance between $g$ and \textsc{Approximate-$g$}.

\begin{cl}\label{srt_claim:g-approxquery-g-are-close}
If \textsc{TestAdditivity}{$(f, 3\alpha)$} (\autoref{alg:subroutines_noisy}) accepts with probability at least $1/3$, then
\textsc{Approximate-$g$} is pointwise $6\alpha$-close to $g$, in the ball $\ball(\bm 0,r)$ with high probability; i.e., 
\[\Pr_{\bm p\sim\cD}[|g(\bm p)-\textsc{Approximate-$g$}(\bm p)|\leq 6\alpha(12R/r)\mid \bm p\in \ball(\bm 0,R)]\geq 1-\eps/4.\]
\end{cl}

The following observation directly follows using the triangle inequality, along with the above claim and the notion of $\eta$-approximate queries.

\begin{obs}\label{srt_obs:f-approx-g-close}
If \textsc{Approximate Additivity Tester} does not reject $f$ with probability at least $2/3$, then following the description of \textsc{Approximate Additivity Tester}, we have
\[\Pr_{\bm p\sim\cD}[|f(\bm p)-\textsc{Approximate-$g$}(\bm p)|\leq 750Rn^{1.5}\alpha\mid\bm p\in \ball(\bm 0,R)]\geq 1-\frac{\eps}{4}.\]
\end{obs}

We will invoke them with $\cD=\cN(\bm 0,I_n)$, implying $R=2\sqrt{n}$ (from \autoref{fact:gaussian-concentration}).

\subsection{Analyses of Subroutines}\label{sec:analyses-subroutines}

\begin{algorithm}[ht]
\caption{$\findbucket$($f, B, S$)}\label{alg:findinfbucket}
{\bf Inputs:} $\eta$-approximate query oracle $\widetilde{f}$ for $f: \R^n \rightarrow \R$, Random Bucketing $B = \{B_1,\hdots,B_r\}$ of $[n]$, $\emptyset\neq S \subseteq [r]$\\

{\bf Output:} Either $\emptyset$, or an influential bucket $B_j$ (with $j \in S$) of $f$. $B_j$ is an influential bucket if $x_v$ is an influential variable of $f$ for some $v \in B_j$.

    Let $V \triangleq \{i \in [n] : i \in B_j \mbox{ for some } j \in S\}$.\Comment{The \#variables in $V$ will be $|V|\approx \frac{|S|n}{r}$.}\

    Sample two independent Gaussian vectors $\vec{x},\vec{y} \sim \cN(\bm 0,I_{n})$.\
    
    \If{
    $|\widetilde{f}(\bm x_{V} \bm y_{\overline{V}})-\widetilde{f}(\bm y)|\leq 2\eta$}{
        \Return $\emptyset$.
    }
    \Else{\If{$S = \{j\}$ for some $j$} {
        \Return $B_j$.
    }
    \Else{
        Partition $S$ into two parts: $S^{(L)}$ and $S^{(R)}$ with $1 \leq |S^{(R)}| \leq |S^{(L)}|\leq |S^{(R)}| + 1$.\\
        
        $\mathsf{ret}_L \leftarrow$ $\findbucket$($f$,$B$, $S^{(L)}$)\\
        \If{$\mathsf{ret}_L = \emptyset$} {
            \Return $\findbucket$($f$,$B$, $S^{(R)}$).
        }
        \Else {
            \Return $\mathsf{ret}_L$.
        }
    }}
\end{algorithm}
The subroutine $\findbucket$ is presented in \autoref{alg:findinfbucket}. To analyze it, we need:

\begin{defi}[Influential bucket for linear functions]
    For a function $f: \R^n \rightarrow \R$, a bucket $B \subseteq [n]$ is said to be an \emph{influential bucket} if there exists at least one variable $x_i,i \in B$, which is influential with respect to $f$ (i.e., the value of $f(\bm{x})$ changes with a change in $x_i$). 
    For a linear $f$, with $f(\bm x) = \sum_{i \in [n]} a_i {x}_i$, ${x}_i$ is influential w.r.t $f$ if and only if $a_i \neq 0$, and hence $B$ is an influential bucket iff $a_i \neq 0$ for some $i \in B$.
\end{defi}

\begin{cl}[Correctness of $\findbucket$]\label{claim:close-to-linear-implies-good_new_bucket}
If $f: \R^n \rightarrow \R$, with $f(\vec{x}) = \sum_{i \in [n]} a_i \vec{x}_i$ is given via an $\eta$-approximate query oracle $\tilde f$, where $\eta\leq\frac{1}{100k^2}\min_{k\in[n]:a_k\neq 0}|a_k|$, $B = \{B_1,\ldots,B_r\}$ is a partition of $[n]$, and $\emptyset\neq S \subseteq [r]$, then $\findbucket(f,B,S)$ (\autoref{alg:findinfbucket}) guarantees the following:
    \begin{enumerate}
        \item If none of the buckets $\{B_i, i\in S\}$ are influential, $\findbucket(f, B, S)$ always returns $\emptyset$ and performs exactly $2$ queries to $f$.
        
        \item Otherwise, with probability at least $1-8 \lceil\log |S|\rceil^2/10k^2$, $\findbucket(f,B,S)$ returns $B_j$ for some $j \in S$ which is an influential bucket, and performs $\leq 8 \ceil{\lg(|S|)}^2$ queries to $f$.
    \end{enumerate}
\end{cl}

\begin{proof}
Let $\vec{x}$ and $\vec{y}$ be the Gaussian random vectors sampled by \autoref{alg:findinfbucket}, and $\vec{w} \triangleq  \vec{y}_{\overline{V}} \vec{x}_V\in \R^n$.

(Part 1) Since no $B_j$ is influential for $j \in S$, we have $a_i = 0$ for all $i \in V$, giving us 
\[f(\vec{w}) - f(\vec{y}) = \sum_{j \in \overline{V}} a_j \underbrace{(\vec{w}_j - \vec{y}_j)}_{= \bm{y}_j-\bm{y}_j = 0} + \sum_{j \in V} \underbrace{a_j}_{=0} \underbrace{(\vec{w}_{j} - \vec{y}_{j})}_{=\bm x_j-\bm y_j} = 0.\]
Hence, $|\widetilde f(\bm w)-\widetilde f(\bm z)|=|\widetilde f(\bm w)-f(\bm w)-(\widetilde f(\bm z)-f(\bm z))|\leq \overbrace{|\widetilde f(\bm w)-f(\bm w)|}^{\leq\eta}+\overbrace{|\widetilde f(\bm z)-f(\bm z)|}^{\leq\eta}\leq 2\eta$ by the triangle inequality, ensuring the check in Line $5$ (\autoref{alg:findinfbucket}) always holds. So in this case, $\findbucket(f,B,S)$ will always return $\emptyset$, and make exactly $2$ queries to $f$.

(Part 2) It only remains to prove the second part of the claim (when there exists some $j \in S$ with $B_j$ being influential, i.e., there exists $k \in B_j$ such that $a_k \neq 0$). 
We do so by strong induction on the size of $S$. 

When $|S| = 1$, we have $S = \{j\}$ and $V = B_j$, with $a_k \neq 0$ for some $k \in B_j$. 
Then $f(\vec{w}) - f(\vec{y}) = \sum_{k \in B_j} a_k (\vec{x}_{k} - \vec{y}_{k})$. This implies $\Pr_{\vec{x},\vec{y}\sim\cN(\bm 0,I_{n})}[f(\vec{w}) = f(\vec{z})] = 0$, since $\sum_{k \in B_j} a_k(\bm x_k-\bm y_k) \not\equiv 0$, and any non-zero polynomial over reals vanishes only on a set (of its zeroes) of measure zero.
Moreover,
\begin{align*}
        \Pr_{\bm x,\bm y\sim\cN(\bm 0,I_{n})}[|\widetilde f(\bm w)-\widetilde f(\bm y)|> 2\eta] &\geq \Pr_{\bm x,\bm y\sim\cN(\bm 0,I_{n})}[|f(\bm w)-f(\bm y)|>4\eta]\\ 
        &=1-\Pr_{\bm x,\bm y\sim\cN(\bm 0,I_{n})}\left[\Big|\sum_{k \in B_j} a_k (\vec{x}_{k} - \vec{y}_{k})\Big|\leq 4\eta\right]\\
        &\geq1-\Pr_{\bm x,\bm y\sim\cN(\bm 0,I_{n})}\left[\max_{k \in B_j}|a_k|\Big|\sum_{k \in B_j} (\vec{x}_{k} - \vec{y}_{k})\Big|\leq 4\eta\right] \\
        & = 1-\Pr_{\bm x,\bm y\sim\cN(\bm 0,I_{n})} \left[\frac{| \sum_{k \in B_j}(\vec{x}_{k} - \vec{y}_{k})|}{\sqrt{2|B_j|}}\leq\frac{2\sqrt{2}\eta}{\sqrt{|B_j|}\max_{k \in B_j}|a_k|}\right].
    \end{align*}    
As $\bm x,\bm y \sim\cN(\bm 0,I_n)$ are independent, 
we have $\sum_{k\in B_j}\frac{\bm x_k-\bm y_k}{\sqrt{2|B_j|}}\equiv \sum_{k\in B_j}\frac{x_k-y_k}{\sqrt{2|B_j|}}\in\R[\cup_{k\in B_j}\{x_k,y_k\}]$, 
$\cov[x_i,x_j]=\cov[y_i,y_j]=0$, for all $i\neq j\in B_j$, 
and $\cov[x_i,y_j]=0$, for all $i,j\in S$, giving us
\[\var_{\bm x,\bm y \sim\cN(\bm 0,I_{n})}\left[\sum_{k\in B_j}\frac{\bm x_k-\bm y_k}{\sqrt{2|B_j|}} \right]=\frac{1}{2|B_j|}\sum_{k\in B_j} \left(\var_{x_k\sim\cN(0,1)}[x_k]+\var_{y_k\sim\cN(0,1)}[y_k]\right)=1.\]
Now, applying \autoref{srt_thoe:carbery-wright} on $\sum_{k\in B_j}\frac{\bm x_k-\bm y_k}{\sqrt{2|B_j|}}$, (with $d=1,n=2|B_j|\text{, and }t=0$) we get
    \[\Pr_{\bm x,\bm y \sim\cN(\bm 0,I_{n})} \left[\frac{| \sum_{k \in B_j}(\vec{x}_{k} - \vec{y}_{k})|}{\sqrt{2|B_j|}}\leq\frac{2\sqrt{2}\eta}{\sqrt{|B_j|}\max_{k \in B_j}|a_k|}\right]\leq O\left(\frac{2\sqrt{2}\eta}{\sqrt{|B_j|}\max_{k \in B_j}|a_k|}\right)\ll\frac{1}{10k^2}.\]
    The last inequality follows by the assumption that
    $\eta\leq\frac{1}{100k^2}\min_{k\in[n]:a_k\neq 0}|a_k|\leq\frac{\sqrt{|B_j|}}{100k^2}\max_{k\in B_j}|a_k|$.
    Hence, with probability at least $1-1/10k^2$, \autoref{alg:findinfbucket} will return $S = \{j\}$ in Line 9.

Now, let $|S| = k > 1$. Using a similar argument as in the base case, $\findbucket(f,B,S)$ will return $\emptyset$ with low probability, since $f(\vec{w}) - f(\vec{y}) = \sum_{j \in V} a_j (\vec{x}_j - \vec{y}_j)$ will be $0$ only when $\vec{x}_V - \vec{y}_V$ lies in a $(|V|-1)$-dimensional hyperplane. However,    
    \begin{align*}
        \Pr_{\bm x,\bm y\sim\cN(\bm 0,I_{n})}[|\tilde f(\bm w)\!-\!\tilde f(\bm y)|\!\leq\! 2\eta] 
        &\leq\!\! \Pr_{\bm x,\bm y\sim\cN(\bm 0,I_{n})}[| f(\bm w)\!-\!f(\bm y)|\!\leq\! 4\eta]=\!\! \Pr_{\bm x,\bm y\sim\cN(\bm 0,I_{n})}\!\! \bigg[\bigg|\sum_{j \in V} a_{j}(\vec{x}_{j} - \vec{y}_{j})\bigg|\!\leq\! 4\eta\bigg]\\
        &\leq\Pr_{\bm x,\bm y\sim\cN(\bm 0,I_{n})}\bigg[\max_{j\in V} |a_{j}|\bigg|\sum_{j \in V} (\vec{x}_{j} - \vec{y}_{j})\bigg|\leq 4\eta\bigg]\\
        &\leq\Pr_{\bm x,\bm y\sim\cN(\bm 0,I_{n})}\left[\left|\sum_{j \in V} \frac{\vec{x}_{j} - \vec{y}_{j}}{\sqrt{2|V|}} \right|\leq \frac{2\sqrt{2}\eta}{\sqrt{|V|}\cdot\max_{j\in V} |a_{j}|}\right]\triangleq (\star).
    \end{align*}
    As earlier, we may observe: as $\bm x,\bm y \sim\cN(\bm 0,I_{n})$ are sampled independently, we have
    \[\sum_{j\in V}\frac{\bm x_{j}-\bm y_{j}}{\sqrt{2|V|}}\equiv \sum_{j\in V}\frac{x_{j}-y_{j}}{\sqrt{2|V|}}\in\R[\cup_{j\in V}\{x_{j},y_{j}\}],\]
    and with $\cov[x_i,x_j]=\cov[y_i,y_j]=0$, for all $i\neq j\in V$, and $\cov[x_i,y_j]=0$, for all $i,j\in V$, we get
    \[\var_{\bm x,\bm y\sim\cN(\bm 0,I_{n})}\left[\sum_{j\in V}\frac{x_{j}-y_{j}}{\sqrt{2|V|}}\right]=\frac{1}{2|V|} \sum_{j=1}^{|V|} \left(\var_{x_{j}\sim\cN(0,1)}[x_{j}]+\var_{y_{j}\sim\cN(0,1)}[y_{j}]\right)=1.\]
    Applying \autoref{srt_thoe:carbery-wright} on $\sum_{j\in V}\frac{x_{j}-y_{
j}}{\sqrt{2|V|}}$, (with $d=1,n=2|V|\text{, and }t=0$) we get
    \[(\star)=\Pr_{\bm x,\bm y\sim\cN(\bm 0,I_{n})}\left[\left|\sum_{j \in V} \frac{\vec{x}_{j} - \vec{y}_{j}}{\sqrt{2|V|}} \right|\leq \frac{2\sqrt{2}\eta}{\sqrt{|V|}\cdot\max_{j\in V} |a_{j}|}\right] \leq O\left(\frac{2\sqrt{2}\eta}{\sqrt{|V|}\cdot\max_{j\in V} |a_{j}|}\right)\ll\frac{1}{10k^2}.\]
    The last inequality again follows by
    $\eta\leq\frac{1}{100k^2}\min_{k\in[n]:a_k\neq 0}|a_k|\leq\frac{\sqrt{|V|}}{100k^2}\max_{k\in V}|a_k|$.
   
Thus, with probability at least $1-1/10k^2$, the condition in Line 5 will not hold, ensuring \autoref{alg:findinfbucket} reaches the recursive step (Lines 8--13). 

By construction, $S = S^{(L)} \sqcup S^{(R)}$, with $1 \leq |S^{(L)}| \leq \ceil{|S|/2}$ and $1 \leq |S^{(R)}| \leq \floor{|S|/2}$. Either,
    \begin{enumerate}[label=(\roman*)]
        \item $a_j = 0$ for all $j \in S^{(L)}$, and there exists $S_R \subseteq S^{(R)}:S_R \neq \emptyset$
        and $a_{j} \neq 0$ for all $j \in S_R$, or
        
        \item there exists $S_L \subseteq S^{(L)}$ such that $S_L \neq \emptyset$ and $a_{j} \neq 0$ for all $j \in S_L$.
    \end{enumerate}

    In Case (i), $\findbucket(f,B,S^{(L)})$ will always return $\emptyset$ and perform $2$ queries (by Part 1) and hence the return value of $\findbucket(f,B,S)$ will be $\findbucket(f,B,S^{(R)})$. 
    Using the strong induction hypothesis we get, with probability at least $1-4\lceil\log |S|\rceil/10k^2$, this return value will be $\{j\}$ for some $j \in S_R$, and the total number of queries made will be $\leq 2 + 2 + 4 \ceil{\lg(\floor{|S|/2})} \leq 4 + 4(\ceil{\lg |S|} - 1) = 4 \ceil{\lg |S|}$.

    In Case (ii), again using strong induction hypothesis, with probability at least $1-4\lceil\log |S|\rceil/10k^2$, $\findbucket(f,B,S^{(L)})$ will return $\{j\}$ for some $j \in S_L$, 
    and thus the algorithm will return $\{j\}$ (line 13). 
    The number of queries made will be $\leq 2 + 4 \ceil{\lg(\ceil{|S|/2})} \leq 2 + 4(\ceil{\lg |S|}-\frac{1}{2}) = 4 \ceil{\lg |S|}$. 
    
    We may note, if we run $\findbucket(f, B, S)$ on some set $S$ of bucket-indices with $|S| > 1$, each recursive step (irrespective of whether it falls in case (i) or case (ii) above) will succeed with probability $\geq 1-4\lceil\log |S|\rceil/10k^2$. 
    For the top level call to succeed, all the recursive calls (at most $2 \lceil \lg |S| \rceil$ in number) must also succeed. 
    So, we do a union bound over all the failure events to upper bound the total failure probability by $8 \lceil\log |S|\rceil^2/10k^2$, and the query complexity by $8 \lceil \lg |S| \rceil^2$.
\end{proof}

Now we are ready to describe and analyze the subroutine $\findbuckets$ (\autoref{alg:findinfbuckets}).

\begin{algorithm}[ht]
\caption{$\findbuckets$($\widetilde{f}, B, X$)}\label{alg:findinfbuckets}
{\bf Inputs:} $\eta$-approximate query oracle $\widetilde{f}$ for $f: \R^n \rightarrow \R$, Bucketing $B = \{B_1,\ldots,B_r\}$, $\emptyset\neq X \subseteq [r]$.

{\bf Output:} A set of  influential buckets $\mathsf{InfBuckets} \subseteq \{B_j : j \in X\}$ with respect to $f$.

    $\mathsf{InfBuckets} \gets \emptyset$ \

\For{j=1 to 8k}{

    $\mathsf{RetVal} = \findbucket(f,B,X)$ \

    \If{$\mathsf{RetVal} = B_i$ for some $i$} {
        $\mathsf{InfBuckets} \gets \mathsf{InfBuckets} \cup \{B_i\}$ \
        
        $X\gets X \setminus \{i\}$ \
    }

}

\Return $\mathsf{InfBuckets}$. \
\end{algorithm}

\begin{cl}[Correctness of $\findbuckets$]\label{claim:close-to-linear-implies-good:findInfbuckets}
If $f: \R^n \rightarrow \R$, with $f(\vec{x}) = \sum_{i \in [n]} a_i \vec{x}_i$, is given via the $\eta$-approximate query oracle $\tilde f$, where $\eta\leq\frac{1}{100k^2}\min_{k\in[n]:a_k\neq 0}|a_k|$, and $\emptyset\neq X \subseteq [n]$, then $\findbuckets(f,B,X)$ (\autoref{alg:findinfbuckets}) performs at most $64 k \lceil \log (|X|)\rceil^2$ queries, and guarantees:
    \begin{enumerate}
        \item[(i)] If $f$ is $\ell$-linear function for some $\ell > 8k$, then \autoref{alg:findinfbuckets} will return a set of $8k$ influential buckets in $f$ with probability at least $1-64 \lceil\log |X|\rceil^2/10k$.

        \item[(ii)] If $f$ is $\ell$-linear function for some $\ell \leq 8k$, then \autoref{alg:findinfbuckets} will return the set of all influential buckets in $f$ with probability at least $1-64 \lceil\log |X|\rceil^2/10k$.

    \end{enumerate}
\end{cl}

\begin{proof}

Note, \autoref{alg:findinfbuckets} calls \autoref{alg:findinfbucket}. So, we use \Cref{claim:close-to-linear-implies-good_new_bucket}, and proceed case-wise:
\begin{itemize}
    \item[(i)] Consider the case when $f$ is $\ell$-linear for some $\ell > 8k$. 
    This implies that the set of indices $X$ given as input to \autoref{alg:findinfbuckets} contains more than $8k$ influential variables. 
    Thus, in every iteration of the \textsc{for} loop starting in Line $4$ of \autoref{alg:findinfbuckets}, \autoref{alg:findinfbucket} in Line $5$ will return an influential bucket, say $B_i$, with probability at least $1-8 \lceil\log |X|\rceil^2/10k^2$, following \Cref{claim:close-to-linear-implies-good_new_bucket} (ii). 
    Line 7 then computes $\mathsf{Val}=a_i$, followed by Line 9 updating $f$ to $f-a_ix_i$, and $X$ to $X\setminus\{i\}$, removing $\{i\}$ from any future iterations. 
    Since the \textsc{for} loop in \autoref{alg:findinfbuckets} runs for $8k$ iterations, and the number of influential variables is more than $8k$, using a union bound over all the iterations, with probability at least $1-64 \lceil\log |X|\rceil^2/10k$, \autoref{alg:findinfbuckets} returns a set of $8k$ influential variables.

    \item[(ii)] When $f$ is $\ell$-linear for some $\ell \leq 8k$. As before, this implies that the set of indices $X$ given as input to \autoref{alg:findinfbuckets} contains at most $8k$ influential variables. 
    Following the same argument, as in the above case, we may claim, the $8k$ iterations of the \textsc{for} loop in \autoref{alg:findinfbuckets} returns all the $\ell\leq 8k$ influential variables in $f$, with probability at least $1-64 \lceil\log |X|\rceil^2/10k$.
\end{itemize}
\emph{Query Complexity:} Since each call of $\findbucket(f,B,X)$ makes $\leq 8 \ceil{\lg(|X|)}^2$ queries to $f$, and there are at most $8k$ such calls, the total number of queries to $f$ is $\leq 64k\ceil{\lg(|X|)}^2$. 
\end{proof}

\subsection{Analysis of the \texorpdfstring{$k$}{k}-linearity tester (\autoref{srt_balg:testklinear})}\label{sec:k-lin-tester-analysis}

We are now ready to prove the main theorem of this section:

\begin{proof}[Proof of \autoref{srt_theo:k_linearity}]

\textbf{Completeness:}
Since $f$ is a $k$-linear function, following \autoref{srt_thm:additivity-tester-arora2}, we have: \textsc{Approximate Additivity Tester} accepts $f$, and hence by \autoref{srt_lem:g-additive-in-ball}, $g$ is pointwise $42\eta$-close to linearity in $\ball(\bm 0,r)$. 
Combined with \autoref{lem:extrapolation_approx}, we get $g$ is pointwise $42\eta(12R/r)$-close to linearity in $\ball(\bm 0,R)$. 
Now from \autoref{srt_claim:g-approxquery-g-are-close}, we know that $g$ and \textsc{Approximate-$g$} are pointwise $6\eta(12R/r)$-close in $\ball(\bm 0,R)$ with probability at least $1-\eps/4$. Using the triangle inequality, this implies that \textsc{Approximate-$g$} is pointwise $48\eta(12R/r)$-close to some linear function with probability at least $1-\eps/4$. Moreover, from \autoref{srt_obs:f-approx-g-close}, we get: \textsc{Approximate-$g$} is in fact pointwise $750Rn^{1.5}\eta$-close to $f$, with probability at least $1-\eps/2$.

With $R=2\sqrt{n}$, following the guarantee of \autoref{claim:close-to-linear-implies-good:findInfbuckets}, which we ensure by our assumption on $\eta$:
\[O(n^{2}\eta)\leq\frac{1}{100k^2}\min_{k\in[n]:a_k\neq 0}|a_k|\text{, or equivalently } \eta\leq O\left(\min_{\substack{i\in[n]:f(\bm e_i)\neq 0}}\frac{|f(\bm e_i)|}{(nk)^2}\right),\]
\color{black}
we get that $\findbuckets(\textsc{Approximate-$g$},B,[r])$ will return at most $k$-influential variables of $f$ with probability at least $1-128 \lceil\log k\rceil^2/5k$.
Thus with probability at least $1-\eps/2-128 \lceil\log k\rceil^2/5k$, \textsc{Test-$k$-Linear} will \textsc{Accept}.

\textbf{Soundness:}
Let us consider the case when $f$ is $\eps$-far from $k$-linearity. We will prove the contrapositive. We will show that if \testlin does not reject $f$ with probability at least $1-\delta(\geq 2/3)$, then $f$ is pointwise close to some $k$-linear function, with non-zero probability.

Note that if \textsc{Approximate Additivity Tester} rejects $f$ with probability $\geq 1-\delta$, we are done. So, let us consider the case when \textsc{Approximate Additivity Tester} accepts with probability $\geq\delta$.
As in the completeness proof, from \autoref{srt_lem:g-is-additive-arora2}, and \autoref{lem:extrapolation_approx}, 
we know that $g$ is pointwise $42\eta(12R/r)$-close to linearity in $\ball(\bm 0,R)$ with probability at least $1-12/125$, and from \autoref{srt_claim:g-approxquery-g-are-close}, 
we have that $g$ and \textsc{Approximate-$g$} are pointwise $6\eta(12R/r)$-close in $\ball(\bm 0,R)$ with probability at least $1-\eps/4$. 
This implies that \textsc{Approximate-$g$} is pointwise $48\eta(12R/r)$-close to some linear function with probability at least $1-\eps/4 - 12/125$. 
Again, from \autoref{srt_obs:f-approx-g-close}, we get: \textsc{Approximate-$g$} is pointwise $750Rn^{1.5}\eta$-close to $f$, with probability at least $1-\eps/4$, implying now $f$ must be pointwise $(48(12R/r)+750Rn^{1.5})\eta$-close to linearity, with probability at least $1-\eps/2 - 12/125$. Note, here $R=2\sqrt{n}$, and $r=1/50$.

With our assumption on $\eta$ again ensuring the conditions for \autoref{claim:close-to-linear-implies-good:findInfbuckets} are met, i.e.,
\[(48(12R/r)+750Rn^{1.5})\eta\leq\frac{1}{100k^2}\min_{k\in[n]:a_k\neq 0}|a_k|,\]
\color{black}
we get  $\findbuckets(\textsc{Approximate-$g$},B,[r])$ will return at most $8k$-influential variables of $f$ with probability at least $1-128 \lceil\log k\rceil^2/5k$.

Since \testlin accepts $f$ with probability $\geq\delta$, this implies that the total number of influential variables returned by $\findbuckets$ is at most $k$, with probability $\geq\delta$.  Combining the above, we can conclude that $f$ is pointwise $(48(12R/r)+750Rn^{1.5})\eta$-close to a $k$-linear function, with probability at least $\delta- 12/125 -\eps/2-128 \lceil\log k\rceil^2/5k$. This concludes the soundness argument.

\textbf{Query complexity:} From \autoref{srt_thm:additivity-tester-arora2}, we know that \textsc{Approximate Additivity Tester} performs $O(\frac{1}{\varepsilon})$ queries. Following \autoref{claim:close-to-linear-implies-good:findInfbuckets}, we also know that $\findbuckets$ performs $\widetilde{O}(k \log k)$ queries. Combining them, we have: \testlin performs $\widetilde{O}(k \log k + \nicefrac{1}{\varepsilon})$ queries in total.
\end{proof}

\section{\texorpdfstring{$k$}{k}-Sparse Low Degree Testing}\label{sec:ksparsepolynomialtest}

In this section, we present and analyze our sparse low degree tester (\autoref{srt_alg:testksparse}).

\begin{theo}[Generalization of \autoref{srt_thm:k-sparse-intro}]\label{srt_theo:ksparsefinalmain}
Let $\eta < \min\{\eps, 1/2^{2^n}\}$. Given $\eta$-approximate query access to $f: \R^n \rightarrow \R$ that is bounded in $\ball(\bm 0,2d\sqrt{n})$, there exists a tester \testsparse(\autoref{srt_alg:testksparse}) that performs $\widetilde{O}(d^5 + \nicefrac{d^2}{\eps} + d k^3)$ queries and guarantees:
\begin{itemize}
    \item \textsc{Completeness:} If $f$ is a $k$-sparse, degree-$d$ polynomial, then \autoref{srt_alg:testksparse}  \textsc{Accepts} with probability at least $1-\eps/4$.
    
    \item \textsc{Soundness:} If $f$ is $\eps$-far from all $k$-sparse degree-$d$ polynomials, then \autoref{srt_alg:testksparse} \textsc{Rejects} with probability at least $2/3$.
    
\end{itemize}
\end{theo}

\begin{algorithm}[h]
\caption{{\testsparse}}\label{srt_alg:testksparse}
{\bf Inputs:} $\eta$-approximate query oracle $\widetilde{f}$ for $f: \R^n \rightarrow \R$, that is bounded in $\ball(\bm 0,R)$, sparsity parameter $ k\in\N$, proximity parameter $\eps \in (0,1)$, degree parameter $d\in\N$.\

{\bf Output:} \textsc{Accept} iff $f$ is $k$-sparse function. \

Run \textsc{ApproxLowDegreeTester}($f,d,\cN(\bm 0,I_n),\eta,\eps,\lrad,\lrad$) (\autoref{alg:low_degree_main_algorithm_approx}). \

\If{\autoref{alg:low_degree_main_algorithm_approx}  rejects}{ 

\Return \textsc{Reject}.\

}

Call $\textsc{Approx-Poly-Sparsity-Test}(\apxqueryg)$ (\autoref{alg:approx-polynomial-sparse-test}). \

\If{\textsc{Approx-Poly-Sparsity-Test}($\apxqueryg$) accepts}
{
\Return \textsc{Accept}.\Comment{$\apxqueryg$ in \autoref{alg:subroutines_approx}}
}

\Else{
\Return \textsc{Reject}.
}

\end{algorithm}

We will prove \autoref{srt_theo:ksparsefinalmain} in \autoref{sec:k-sparsity-tester-analysis}, after developing the necessary machinery. 
\autoref{srt_alg:testksparse} first invokes \apxlowdegtest (\autoref{alg:low_degree_main_algorithm_approx}) to reject functions that are far from any low-degree polynomial.
We record some useful claims about \autoref{alg:low_degree_main_algorithm_approx}, and its subroutines (\autoref{alg:subroutines_approx}) from \cite{arora2} and (its improvement in) \cite{arora2024optimaltestinglinearitysosa}:

\begin{restatable}[{\cite[Theorem 3.6]{arora2024optimaltestinglinearitysosa}}]{theo}{approximateLowDegree}\label{thm:approximate_low_degree_tester}
Let $d\in\mathbb{N}$, for $L>0$, $f: \R^n \to \R$ be bounded in the ball $\ball(\bm{0}, L)$, and for $\eps \in (0,1), R>0$, let $\cD$ be an $(\eps/4, R)$-concentrated distribution.
For $\alpha>0, \beta\geq 2^{(2n)^{O(d)}}(R/L)^d\alpha$, given $\alpha$-approximate query access to $f$, and sampling access to $\mathcal{D}$, there is an one-sided error, $O(d^5+\frac{d^2}{\varepsilon})$-query \apxlowdegtest (\autoref{alg:low_degree_main_algorithm_approx}) which, distinguishes between the case when  $f$ is pointwise $\alpha$-close to some degree-$d$ polynomial and the case when, for every degree-$d$ polynomial $h\colon\mathbb{R}^n\to\mathbb{R}$, $\Pr_{\bm p \sim \cD}[|f(\bm p)-h(\bm p)| > \beta ] > \eps$.
\end{restatable}

\begin{lem}[{\cite[Lemma 4.4]{arora2}}] \label{lem:main_approx_poly}
    Let $\srad=(4d)^{-6}$, $\delta=2^{d+1}\alpha$, as  set in \autoref{alg:subroutines_approx}, and $\lrad>\srad$. 
    If \textsc{ApproxCharacterizationTest} fails with probability at most $2/3$, then $g$ is pointwise $2^{(2n)^{45d}}(\lrad/L)^d\delta$-close to a degree-$d$ polynomial in $\ball(\bm 0,2d\lrad\sqrt{n}/L)$.
    Furthermore, for every point $\bm p \in \ball(\bm 0,2d\lrad\sqrt{n}/L)$, \textsc{ApproxQuery-$g$}($\bm{p}$) well approximates $g(\bm p)$ with high probability, that is, 
    \[\Pr_{\bm p\sim\cD}\Big[|g(\bm{p})-\text{\textsc{ApproxQuery-$g$}($\bm{p}$})|\leq\Big(\frac{24d\lrad\sqrt{n}}{L\srad}\Big)^d2^{d+4}\delta\Big]\geq 1-\frac{\eps}{4}.\]   
\end{lem}
To invoke these results, we assume: (i) $f$ is bounded in $\ball(\bm 0, R)$, i.e., we set $L=R$, and (ii) $\alpha=\eta$. Additionally, since we work over standard Gaussians, we set $R=2d\sqrt{n}$.

\subsection{Testing sparsity of polynomials given exact query access}\label{sec:sparsity-poly-exact}

In this section, we design an algorithm for testing sparsity of \emph{polynomial functions} $f: \R^n \rightarrow \R$, by extending the machinery developed in \cite{gjr10} and \cite{bot88} to the real numbers.

Recall the notion of Hankel Matrices associated with polynomials (\autoref{srt_defi:Hankel-matrix-poly}).
\Hankelmatrixpoly*
We note the following observation about Hankel matrices. It follows essentially the same argument as in \cite{bot88,gjr10}, since the decomposition that they use over finite domains also works over the reals. For completeness, a proof of this observation is provided in \Cref{sec:omitted-proofs}.

\begin{restatable}[Generalization of {\cite[Section 4]{bot88}}, and {\cite[Lemma 4]{gjr10}}]{obs}{generalizedhenkelmatrix}\label{obs:Hankel-charcterization}
	Let $f:\R^n\to\R$ be an exactly $k$-sparse polynomial over the reals, i.e., $f(\bm x)=\sum_{i=1}^k a_i M_i(\bm x)$, where $a_1,\ldots,a_k \in \R \setminus \{0\}$, and $M_1,\ldots,M_k$ are the monomials of $f$. Then for all $\ell+1 \leq k$,
	\[\det \paren{H_{\ell+1}(f,\vec{x})}= \sum_{\substack{ S\subseteq [k]\\ |S|=\ell+1}} \prod_{i\in S} a_i \prod_{\substack{i,j\in S \\ i<j}} (M_j(\vec{x}) - M_i(\vec{x}))^2,\]
	is a non-zero polynomial of degree $\leq 2 \binom{\ell+1}{2} \deg(f)$ in $\vec{x}$, while for all $\ell+1>k$, $\det \paren{H_{\ell+1}(f,\vec{x})} \equiv 0$.
\end{restatable}

As a preliminary, we argue that \Cref{obs:Hankel-charcterization} can be used to test whether a polynomial $f$ is $(\leq k)$-sparse, or not (with error probability $0$, in fact!), given exact query access to $f$. 

\begin{lem}\label{srt_lem:perfect-sparsity-tester}
Let $f: \R^n \rightarrow \R$ be a polynomial function (of any degree). Given exact query access to $f$, there is an algorithm that makes $2k+1$ queries to $f$, and exactly tests whether $\|f\|_0 \leq k$ (returns \textsc{Accept}), or $\|f\|_0 > k$ (returns \textsc{Reject}) with error probability $0$.  
\end{lem}

\begin{proof}
The algorithm is as follows: Sample a single point $\vec{u} \sim \cN(\bm 0, I_n)$. Construct the Hankel matrix $H_{k+1}(f,\vec{u})$ as in \Cref{srt_defi:Hankel-matrix-poly} using $2k-1$ exact queries to $f$. Test whether $\det\paren{H_{k+1}(f,\vec{u)}} \stackrel{?}{=} 0$ and return \textsc{Accept} if the determinant is $0$, and \textsc{Reject} otherwise.
	 
Let $Q(\vec{x}) \equiv \det (H_{k+1}(f,\bm x)) \in \R[\vec{x}]$, so that the test is $Q(\vec{u}) \stackrel{?}{=} 0$. 
If $\|f\|_0 \leq k$, by \Cref{obs:Hankel-charcterization}, 
$Q$ is the zero polynomial and hence the algorithm will always \textsc{Accept} after finding that $Q(\vec{u}) = 0$.
If $\|f\|_0 \geq k + 1$, then $Q$ is a non-zero polynomial, and hence $\Pr_{\vec{u} \sim \cN(\bm 0,I_n)}[Q(\vec{u}) = 0] = 0$; this is because the Lebesgue measure, and hence the probability measure with respect to any continuous distribution, of the zero-set of any non-zero real polynomial is $0$. 
Thus, the algorithm will \textsc{Reject} with probability $1$ over the randomness of $\vec{u}$.
\end{proof}

\subsection{Testing sparsity of polynomials given approximate query access}\label{sec:sparsity-poly-approx}
Now, instead of exact query access, we have $\eta$-approximate query access to a polynomial $f$. 
We make the following observation about the Hankel matrix constructed with the queried values $\tilde{f}$, since the approximate-query oracle guarantees that $|\tilde{f}(\vec{z}) - f(\vec{z})| \leq \eta$ for all points $\vec{z} \in \R^n$. 
We provide a proof of this observation in \Cref{sec:omitted-proofs}.

\begin{restatable}{obs}{approxhankelstructure} \label{obs:approx-hankel-structure}
	Let $f: \R^n \rightarrow \R$ and let $\tilde{f}$ be an $\eta$-approximate query oracle to $f$. Then for any $t \geq 1$ and $\vec{u} \in \R^n$, we can express
	\[
		H_t(\tilde{f},\vec{u}) = H_t(f,\vec{u}) + E_t(\vec{u}),
	\]
	where $E_t(\vec{u})$ is a Hankel-structured, \emph{noise matrix}, such that $\|E_t(\vec{u})\|_\infty \leq \eta$, and $\|E_t(\vec{u})\|_{\rm op} \leq \eta t$.
\end{restatable}

We also need the following probabilistic upper bound on the eigenvalues of the Hankel matrix of a polynomial, that was briefly introduced in \autoref{sec:k-sparse-intro}. We restate it here for convenience.

\upperboundonsigmamax*

\begin{proof}[Proof of \autoref{lem:upperboundonsigmamax}]
Let $H_{\vec{u}}$ denote $H_t(f,\vec{u})$. Then for any $\vec{z}=(z_1,\hdots,z_t)^\top \in \R^t$, we have
\begin{align*}
    |\vec{z}^\top H_{\vec{u}} \vec{z}| &= \left|\sum_{i \in [t]} \sum_{j \in [t]} [H_{\vec{u}}]_{i,j} z_i z_j\right| \leq \sum_{i \in [t]} \sum_{j \in [t]} \left|f(\vec{u}^{i+j-2})\right| |z_i| |z_j| \tag{$\because [H_{\bm u}]_{i,j}=f(\bm u^{i+j-2})$}\\
    &\leq\! \sum_{i \in [t]} \sum_{j \in [t]} |z_i| |z_j| \left(\sum_{p \in [k]}\! |a_p| |M_p(\vec{u}^{i+j-2})|\right) \!\!= \vec{\overline{z}}^\top V^\top
		D V
 \vec{\overline{z}},\mbox{ where }\overline{\bm z}=(|z_1|,\hdots,|z_t|)^\top\!\!\in\!\R_{\geq 0}^t,\\
 V &= \begin{pmatrix}
				1 & 1 & \hdots & 1 \\
				|M_1(\bm u)| & |M_2(\bm u)| & \hdots & |M_k(\bm u)| \\
				|M_1(\bm u)|^2 & |M_2(\bm u)|^2 & \hdots & |M_k(\bm u)|^2 \\
				\vdots & \vdots & \ddots & \vdots \\
				|M_1(\bm u)|^{t-1} & |M_2(\bm u)|^{t-1} & \hdots & |M_k(\bm u)|^{t-1}
		\end{pmatrix}^\top, 
        \mbox{ and }
D = \begin{pmatrix}
			|a_1| & 0 & \hdots & 0 \\
			0 & |a_2| & \hdots & 0 \\
			\vdots & \vdots & \ddots & \vdots \\
			0 & 0 & \hdots & |a_k|
		\end{pmatrix},
\end{align*}
wherein in the last step, we use the same decomposition as in the proof of \Cref{obs:Hankel-charcterization}, but applied to $H_t(|f|,\vec{u})$, with $|f|(\vec{u}) \triangleq \sum_{p \in [k]} |a_p| |M_p(\vec{u}^{i+j-2})|$.

Thus, if $\|\vec{\overline{z}}\|_2 = \|\vec{z}\|_2 \leq 1$, we have
\[
\sigma_{\max}(H_{\vec{u}}) \leq \vec{\overline{z}} V^\top (D^{\frac12})^\top D^{\frac12} V \vec{\overline{z}} = \left\|D^{\frac12} V \vec{\overline{z}}\right\|_2^2 \leq \|D^{\frac12} V\|_{\rm op}^2 \|\vec{\overline{z}}\|_2^2 \leq \|D^{\frac12} V\|_{\rm op}^2 \leq \|D^{\frac12} V\|_{\rm F}^2\text{, where}
\]
\[D^{\frac12} V = 
\begin{pmatrix}
|a_1| & |a_1| |M_1(\vec{u})| & \cdots & |a_1| |M_1(\vec{u})|^{t-1}\\
|a_2| & |a_2| |M_2(\vec{u})| & \cdots & |a_2| |M_2(\vec{u})|^{t-1}\\
\vdots & \vdots & \ddots & \vdots\\
|a_k| & |a_k| |M_k(\vec{u})| & \cdots & |a_k| |M_k(\vec{u})|^{t-1}
\end{pmatrix}.\] 
By $U \geq \max\left\{2, \maxop_{p \in [k]} |M_p(\vec{u})|\right\}$, we get 
$\|D^{\frac12} V\|_{\rm F}^2 \leq \sum_{r=0}^{t-1} \|\vec{a}\|_2^2 U^{2r} = \|\vec{a}\|_2^2 \left(\frac{U^{2t}-1}{U-1}\right) \leq \|\vec{a}\|_2^2 \cdot U^{2t}$.

Finally, for any $\bm \alpha\triangleq(\alpha_1,\ldots,\alpha_n) \in \N^n$, we have $\E_{\vec{u} \sim \cN(\bm 0,I_n)}[|\vec{u}^{\bm \alpha}|] = \prod_{i \in [n]}\E_{u_i \sim \cN(0, 1)}[|u_i^{\alpha_i}|]$. 
From \cite[Equations (1), (4), and (15)]{elandt1961folded}, for $s \in \N_{>0}$, we get
\begin{equation}
    \E_{u \sim \cN(0,1)}[|u^{s}|] =\begin{cases}
        \frac{s!}{2^{s/2} (s/2)!} \leq \frac{2^{s/2} \lfloor s/2 \rfloor!}{\sqrt{\pi}} &\text{, if $s$ is even (and $> 0$), and }\\
        2^{\lfloor s/2 \rfloor} \left(\lfloor \frac{s}{2} \rfloor\right)!\sqrt{\frac{2}{\pi}} = \frac{2^{s/2} \lfloor s/2 \rfloor!}{\sqrt{\pi}} &\text{, if $s$ is odd.}
    \end{cases}
\end{equation}

The inequality in the ``$s$ is even'' case follows from: (i) $\frac{s!}{2^{s}(s/2)!} = \frac{\Gamma\left(\frac{s}{2} + \frac{1}{2} \right)}{\sqrt{\pi}}$ \cite{wolframalpha_eq1}, and (ii) the monotonicity of the Gamma function in $(\frac12, \infty)$ which implies $\Gamma\left(\frac{s}{2} + \frac12\right) \leq \Gamma\left(\frac{s}{2} + 1\right) = \frac{s}{2}!$ for even $s$.

Since $\E_{u \sim \cN(0,1)}[|u^0|] = 1$, for any $\bm \alpha = (\alpha_1,\ldots,\alpha_n) \in \N^n$ with $\|\bm \alpha\|_1 = \sum_{i \in [n]} \alpha_i \leq d$, we have 
\[\E_{\vec{u} \sim \cN(\bm 0,I_n)}[|\vec{u}^{\bm \alpha}|] \leq \prod_{i : \alpha_i \neq 0} \left(\frac{2^{\alpha_i/2}\lfloor \alpha_i /2 \rfloor!}{\sqrt{\pi}}\right) \leq 2^{d/2} \lceil d/2 \rceil !.\]
By a similar argument, $\mathrm{Var}[|\vec{u}^{\bm \alpha}|] \leq \E_{\vec{u}}[|\bm u^{2\bm\alpha}|] \leq 2^d \cdot d!$. 

Thus, by Chebyshev's inequality, for any $d \geq 0$, $k \geq 1$, and $\bm \alpha \in \N^n$ with $\|\bm \alpha\|_1 \leq d$,
\[\Pr_{\vec{u} \sim \cN(\bm 0,I_n)}\left[|\vec{u}^{\bm \alpha}| \geq 2^{d/2} \lceil d/2 \rceil ! + \sqrt{\frac{k}{\gamma}} 2^{d/2} \sqrt{d!}\right] \leq \frac{\gamma}{k}.\]

Thus, from the above discussion, $M_p(\bm{u}) \leq \left(2^{d/2} \lceil d/2 \rceil ! + \sqrt{\frac{k}{\gamma}} 2^{d/2} \sqrt{d!}\right)$ for every $p \in [k]$, w.p $\geq 1 - \gamma$ (using the union bound), and hence $\sigma_{\max}(H_t(f,\bm{u})) \leq \|\bm{a}\|_2^2 \left(2^{d/2} \lceil d/2 \rceil ! + \sqrt{\frac{k}{\gamma}} 2^{d/2} \sqrt{d!}\right)^{2t}$.
\end{proof}

We now present and analyze the correctness of \autoref{alg:approx-polynomial-sparse-test}:

\begin{algorithm}
	\caption{\textsc{Approx-Poly-Sparsity-Test}: Approx-query sparsity test for polynomials}\label{alg:approx-polynomial-sparse-test}
	{\bf Input:} $\eta$-approximate query oracle $\tilde{f}$ to a polynomial $f:\R^n\to\R$ of total degree $\leq d$, sparsity parameter $k \in \N$. \\

    {\bf Output:} \textsc{Accept} iff $f$ is a $k$-sparse polynomial. \\
 
	Let $T \gets 4\min\{C,1\}\cdot d(k+1)^2$, where $C$ is the absolute constant in \Cref{srt_thm:glazer-mikulincer}.\\
	\For{$t \in [T]$}{
		Sample $\vec{u}_t = (u_{t,1},\ldots,u_{t,n}) \sim \cN(\bm 0, I_n)$.\\
		\For{$i = 0,\ldots,2k$}{
			Compute $\vec{u}_t^i\triangleq (u_{t,1}^i,\hdots,u_{t,n}^i)$, and $\tilde{f}(\vec{u}_t^i)$.\\
		}
		Compute the Hankel matrix $H_{k+1}(\tilde{f},\vec{u}_t)$ as in \Cref{srt_defi:Hankel-matrix-poly}.\\
		Compute the smallest singular value $\sigma_{\min}^{(t)}$ of $H_{k+1}(\tilde{f},\vec{u}_t)$.\\
	}
	\If{$\sigma_{\min}^{(t)} \leq \eta(k + 1)$ for all $t \in [T]$} {
		\Return \textsc{Accept} (($\leq k$)-sparse).\\
	} 
	\Else {
		\Return \textsc{Reject} (($>k$)-sparse).\\
	}
\end{algorithm}

\begin{theo}[Sparsity testing of polynomials with approximate queries]\label{lem:approx-query-sparsity-testing-poly}
Given $\eta$-approximate query access $\widetilde{f}$ to a polynomial $f: \R^n \rightarrow \R$ of total degree $\leq d$, assuming $\eta \leq \frac{\left(\mathrm{coeff}_{d_Q}(Q)\right)^{\frac{1}{2}}}{(2(k+1))2^{\Theta(k^3 d)}\left(\sigma_{\max}(H_{\vec{u}})\right)^{k}}$, 
wherein $H_{\bm u}=H_{k+1}(f,\bm u), Q_{\bm u}=(\det(H_{\bm u}))^2, d_Q=\deg(Q)$, and  $\mathrm{coeff}^2_{d_Q}(Q)$ is as in \autoref{srt_thm:glazer-mikulincer}, {\rm\textsc{Approx-Poly-Sparsity-Test}}  (\autoref{alg:approx-polynomial-sparse-test}) performs at most {$O(dk^3)$} queries and guarantees:
\begin{enumerate}
	\item[(i)] If $f$ has sparsity at most $k$, the algorithm always \textsc{Accepts}, and
 
	\item[(ii)] If $f$ has sparsity $> k$, the algorithm \textsc{Rejects} with probability at least $\frac{2}{3}$.
\end{enumerate} 
\end{theo}

\begin{proof}
The query complexity follows directly from \autoref{alg:approx-polynomial-sparse-test}, since the tester performs at most $(2k+1)T$ queries for $T \leq O(k^2 d)$. It remains to argue the completeness and soundness guarantees.
 
For any $\vec{u} \in \R^n$, let us denote $H_{k+1}(\tilde{f}, \vec{u})$ by $\Tilde{H}_{\bm u}$ and $H_{k+1}(f, \vec{u})$ by $H_{\bm u}$. Observe, since $\Tilde{H}_{\vec{u}}$ and $H_{\vec{u}}$ are symmetric matrices, their singular values are just the absolute values of their eigenvalues, i.e., $\sigma_{\min}(\Tilde{H}_{\vec{u}}) = \min_{i \in [k+1]}|\lambda_{i}(\Tilde{H}_{\vec{u}})|$ and $\sigma_{\min}(H_{\vec{u}}) = \min_{j \in [k+1]} |\lambda_{j}(H_{\vec{u}})|$.

\paragraph*{Completeness:}We have $\|f\|_0 \leq k$. 
By \Cref{obs:Hankel-charcterization}, we get $\det\paren{H_{k+1}(f, \vec{x})}\equiv 0$, giving us: for any $\vec{u} \in \R^n$, $\det\paren{H_{\vec{u}}} = 0$, and hence, $\lambda_{i^\ast}(H_{\vec{u}}) = 0$ for some $i^\ast \in [k+1]$. 
From \Cref{obs:approx-hankel-structure}, we can write $\Tilde{H}_{\vec{u}} = H_{\vec{u}} + E$, for a symmetric matrix $E$ with $\|E\|_{\rm op} \leq \eta(k + 1)$. 
Then, by Weyl's inequality (\Cref{srt_thm:weyls-ineq}), we have $|\lambda_{i^\ast}(\Tilde{H}_{\vec{u}}) - \lambda_{i^\ast}(H_{\vec{u}})| \leq \eta(k + 1)$, which implies $\sigma_{\min}(\Tilde{H}_{\vec{u}}) \leq |\lambda_{i^\ast}(\Tilde{H}_{\vec{u}})| \leq \eta(k + 1)$. 
So, \autoref{alg:approx-polynomial-sparse-test} will always \textsc{Accept} (Lines 10--11).

Observe: For any non-singular symmetric matrix $A \in \R^{(k+1)\times (k+1)}$, we have 
\[ \left(\sigma_{\min}(A)\right)^{2(k+1)} \leq (\det(A))^2 = \prod_{i=1}^{k+1} (\sigma_i(A))^2 \leq (\sigma_{\min}(A))^{2} (\sigma_{\max}(A))^{2k},\]
where $\sigma_{\max}(A)$ is the largest magnitude eigenvalue of $A$ and $\sigma_{\min}(A)$ is the smallest magnitude eigenvalue of $A$. Thus, $\sigma
_{\min}(A) \leq \Xi$ implies $\det(A)^2 \leq (\Xi)^{2}\sigma_{\max}(A)^{2k}$.

\paragraph*{Soundness:} We have $\|f\|_0 > k$. 
Let $Q(\vec{x}) \triangleq \paren{\det\paren{H_{k+1}(f,\vec{x})}}^2$. 
By \Cref{obs:Hankel-charcterization}, $Q$ is a non-zero polynomial with total degree $\leq 2 (k+1)^2 d$, and $(\det(H_{\vec{u}}))^2 = Q(\vec{u})$ for all $\vec{u} \in \R^n$.

Let $\vec{u} \sim \cN(\bm 0, I_n)$, and $\sigma_{\min}(\Tilde{H}_{\vec{u}}) \leq \eta(k+1)$. 
Then, as in completeness, by Weyl's Inequality (\Cref{srt_thm:weyls-ineq}), and \Cref{obs:approx-hankel-structure}, we have $\sigma_{\min}(H_{\vec{u}}) \leq 2\eta(k+1)$. 
Thus, we must have 
\[Q(\vec{u}) = (\det(H_{\vec{u}}))^2 \leq (2\eta(k+1))^{2} \sigma_{\max}(H_{\vec{u}})^{2k} \triangleq \Delta,\]
as long as $\det(H_{\vec{u}}) \neq 0$ (which will happen with probability $1$ over the choice of $\vec{u}$, since $H_{\bm u}\not\equiv 0$).

First, suppose that $f$ is not a constant polynomial. 
Then, $d_Q \triangleq \deg(Q) \geq 2$ by construction, since it is the square of a non-constant polynomial. 
Now, we may invoke an anti-concentration result.

From \Cref{srt_thm:glazer-mikulincer}, we have $\Pr_{\vec{u} \sim \cN(\bm 0, I_n)}\left[Q(\vec{u}) \leq \Delta\right] \leq Cd_Q\left(\frac{\Delta}{\mathrm{coeff}_{d_Q}(Q)}\right)^{1/d_Q}$. Observe,

\begin{align*}
Cd_Q\left(\frac{\Delta}{\mathrm{coeff}_{d_Q}(Q)}\right)^{1/d_Q} \leq 1 - \frac{1}{C d_Q} \iff \left(\frac{\Delta}{\mathrm{coeff}_{d_Q}(Q)}\right) \leq \left(\frac{1}{C d_Q} - \left(\frac{1}{C d_Q}\right)^2\right)^{d_Q}.
\end{align*}

Assuming $C \geq 1$, so that $1/(C d_Q) \leq 1/2$ (since $d_Q \geq 2$)~\footnote{We renormalize the constant from \Cref{srt_thm:glazer-mikulincer} appropriately.}, we have

\begin{align*}
	\ln\left(\left(\frac{1}{C d_Q} - \left(\frac{1}{C d_Q}\right)^2\right)^{d_Q}\right) &= d_Q \left[\ln\left(1 - \frac{1}{C d_Q}\right) - \ln(C d_Q)\right] \geq d_Q\left[\frac{-1/(C d_Q)}{1 - 1/(C d_Q)} - \ln(C d_Q)\right]\\
	&\geq -d_Q \ln(C d_Q) - \frac{2}{Cd_Q} \geq -(d_Q \ln(C d_Q) + 1),
\end{align*}
where we use the inequality $\ln(1+z) \geq \frac{z}{1+z}$ for $z > -1$, and the fact that $1/(C d_Q) \leq 1/2$. 
Then 
\[\left(\! \frac{1}{C d_Q} \! - \! \left(\! \frac{1}{C d_Q} \! \right)^2 \! \right)^{d_Q} \!\!\!\!\!\!\! \geq \! \exp\left(-(d_Q \ln(C d_Q) + 1)\right) \! \geq \underbrace{\exp\left(-2 (k+1)^2 d \ln(2C (k+1)^2 d) + 1\right)}_{\triangleq \mathsf{UB}_{d,k}} \geq 2^{-\Theta(k^3 d)}.\]

Thus, if $\Delta \leq \mathsf{UB}_{d,k} \cdot \mathrm{coeff}_{d_Q}(Q)$, we will have $\Pr_{\vec{u} \sim \cN(\bm 0,I_n)}[Q(\vec{u}) \leq \Delta] \leq 1 - \frac{1}{C d_Q}$.
If $\eta \leq \frac{\left(\mathrm{coeff}_{d_Q}(Q)\right)^{\frac{1}{2}}}{(2(k+1))2^{\Theta(k^3 d)}\left(\sigma_{\max}(H_{\vec{u}})\right)^{k}} \leq \frac{\left(\mathsf{UB}_{d,k} \cdot \mathrm{coeff}_{d_Q}(Q)\right)^{\frac{1}{2}}}{(2(k+1))\left(\sigma_{\max}(H_{\vec{u}})\right)^{k}}$, then $\Delta = (2\eta(k+1))^{2}\sigma_{\max}(H_{\vec{u}})^{2k} \leq \mathsf{UB}_{d,k} \cdot \mathrm{coeff}_{d_Q}(Q)$.

With the above condition satisfied, we will have
\[
\Pr_{\vec{u}}\left[\sigma_{\min}(\Tilde{H}_{\vec{u}}) \leq \eta(k+1)\right] \leq \Pr_{\vec{u}}\left[\sigma_{\min}(H_{\vec{u}}) \leq 2\eta(k+1)\right] \leq \Pr_{\vec{u}}[Q(\vec{u}) \leq \Delta] \leq 1 - \frac{1}{C d_Q}. 
\]

Now, consider the operation of \autoref{alg:approx-polynomial-sparse-test}. In every iteration $t \in [T]$ for $T = 4 \min\{C, 1\} \cdot d(k+1)^2 \geq 2 C d_Q$, we will have $\Pr_{\vec{u}_t}\left[\sigma_{\min}(\Tilde{H}_{\vec{u}_t}) \leq \eta(k+1)\right] \leq 1 - 1/(C d_Q)$ by the above argument. 
Then, over $T$ independent rounds, the probability that this event occurs every time, i.e., $f$ is accepted, is 
\[\prod_{t \in [T]} \Pr_{\vec{u}_t}\left[\sigma_{\min}(\Tilde{H}_{\vec{u}_t}) \leq \eta(k+1)\right] \leq \left(1 - \frac{1}{C d_Q}\right)^T \leq\frac{1}{e^2} < \frac{1}{3}\]
for our choice of $T$. Thus, the algorithm will \textsc{Reject} with probability at least $\frac{2}{3}$.
\end{proof}

\subsection{Analysis of \texorpdfstring{$k$}{k}-Sparsity Low Degree Tester}\label{sec:k-sparsity-tester-analysis}

We are now ready to analyze our sparse low degree tester (\autoref{srt_alg:testksparse}).

\begin{proof}[Proof of \Cref{srt_theo:ksparsefinalmain}]

Let us start with completeness.

\textbf{Completeness:} Since $f$ is a $k$-sparse,\ degree-$d$ polynomial, 
from \autoref{thm:approximate_low_degree_tester}, we have: 
\textsc{Approx-Low-Degree-Tester} always accepts $f$, and hence, from \autoref{lem:main_approx_poly}, we have that $g$ is pointwise $2^{(2n)^{45d}}R^d 2^{d+1}\eta$-close to a degree-$d$ polynomial, say $h$, in $\ball(\bm 0, R)$. Moreover,
\[\Pr_{\bm p\sim\cD}\Big[|g(\bm{p})-\text{\textsc{ApproxQuery-$g$}($\bm{p}$})|\leq\Big(\frac{24d\sqrt{n}}{\srad}\Big
)^d2^{2d+5}\eta \mid \bm p\in \ball(\bm 0,2d\sqrt{n})\Big]\geq 1-\frac{\eps}{2}.\]
\begin{equation}\label{eq:h-approx-g-close}
    \implies\Pr_{\bm p\sim\cD}\Big[|h(\bm{p})-\text{\textsc{ApproxQuery-$g$}($\bm{p}$})|\leq\left(\Big(\frac{24d\sqrt{n}}{\srad}\Big)^d2^{2d+5}+ 2^{(2n)^{45d}}R^d 2^{d+1}\right)\eta\Big]\geq 1-\frac{\eps}{4}.
\end{equation}

By setting $\eta=1/2^{2^n}$, the conditions of \Cref{lem:approx-query-sparsity-testing-poly} are met, giving us that \textsc{Approx-Poly-Sparsity-Test(ApproxQuery-}$g$) (\autoref{alg:approx-polynomial-sparse-test}) will \textsc{Accept} with probability $\geq 1-\eps/4$. 
So with probability at least $1-\eps/4$, \textsc{Test-$k$-Sparse} (\autoref{srt_alg:testksparse}) will accept $f$.

\textbf{Soundness:} Let $f$ be $\eps$-far from all $k$-sparse, degree-$d$ polynomials. We will show that if \autoref{srt_alg:testksparse} accepts with probability at least $1/3$, then $f$ must be $\eps$-close to some $k$-sparse, degree-$d$ polynomial.
From the premise, \textsc{Approx-Low-Degree-Tester} (\autoref{alg:low_degree_main_algorithm_approx}) also accepts $f$ with probability at least $1/3$.
Then, as in completeness, from \autoref{lem:main_approx_poly}, we have:
$g$ is pointwise $2^{(2n)^{45d}}R^d 2^{d+1}\eta$-close to some degree-$d$ polynomial, say $h(\vec{x}) = \sum_{\|\bm M_i\|_1\leq d} a_i \bm M_i(\vec{x})$, 
i.e., for all $\bm p\in \ball(\bm 0,R)$, $|g(\bm p)-h(\bm p)|\leq 2^{(2n)^{45d}}R^d 2^{d+1}\eta$, and \eqref{eq:h-approx-g-close} still holds.

From the premise, \textsc{Approx-Poly-Sparsity-Test}($\apxqueryg$) (\autoref{alg:approx-polynomial-sparse-test}) does not reject with probability at least $1/3$. In this case, as long as the closeness of $\apxqueryg$ and $h$ satisfies the assumption in \Cref{lem:approx-query-sparsity-testing-poly}, $h$ will be $k$-sparse. The assumption is: 
\begin{equation}\label{eq:assumption1}
    (3\cdot 2^{(2n)^{45d}}+1)\eta\leq \frac{\left(\mathrm{coeff}_{d_Q}(Q)\right)^{\frac{1}{2}}}{(2(k+1))2^{\Theta(k^3 d)}\left(\sigma_{\max}(H_{\vec{u}})\right)^{k}},
\end{equation}
wherein $H_{\bm u}=H_{k+1}(h,\bm u), Q({\bm u})=(\det(H_{\bm u}))^2, d_Q=\deg(Q)$, $\mathrm{coeff}^2_{d_Q}(Q)\leq(\|\bm a\|_2^{2}n^d)^{(k+1)}((k+1)!)^2$, and $\sigma_{\max}(H_{\vec{u}})$ may be bounded using \autoref{lem:upperboundonsigmamax} (with $\gamma=0.01,t=k+1$), i.e., 
\[\Pr_{\vec{u} \sim \cN(\bm 0, I_n)}\left[\sigma_{\max}(H_{\bm u}) \geq \|\vec{a}\|_2^2 \left(2^{d/2} \lceil d/2 \rceil ! + \sqrt{\frac{k}{0.01}} 2^{d/2} \sqrt{d!}\right)^{2(k+1)}\right] \leq 0.01.\]
So, with probability at least $0.99$, $\sigma_{\max}(H_{\bm u}) \leq \|\vec{a}\|_2^2 \left(2^{d/2} \lceil d/2 \rceil ! + 10\sqrt{k} 2^{d/2} \sqrt{d!}\right)^{2(k+1)}$. Plugging these into \eqref{eq:assumption1}, we observe, setting $\eta\leq 1/2^{2^n}$ satisfies it, implying that $f$ is $\eps$-close to some $k$-sparse, degree-$d$ polynomial.

\textbf{Query complexity:} The query complexity of \autoref{srt_alg:testksparse} consists of two parts:
\begin{itemize}
    \item query complexity of \apxlowdegtest (\autoref{alg:low_degree_main_algorithm_approx}) which is $O(d^5+d^2/\eps)$ (from \autoref{thm:approximate_low_degree_tester}), and 
    \item the query complexity of \textsc{Approx-Poly-Sparsity-Test} (\autoref{alg:approx-polynomial-sparse-test}) which is at most $O(dk^3)$ (from \Cref{lem:approx-query-sparsity-testing-poly}).
\end{itemize} 
Combining the above, we can say that \autoref{srt_alg:testksparse} performs $O(d^5 + d^2/\eps + d k^3)$ queries in total.
\end{proof}

\section{\texorpdfstring{$k$}{k}-Junta Testing}\label{sec:kjuntatest}

In this section, we present and analyze our $k$-junta tester.

\begin{theo}[Generalization of \autoref{srt_thm:junta-intro}]\label{theo:kjuntafinalmain}
Let $\eps \in (0,1)$ and $\eta < \min\{\eps/16k^2, O(\nicefrac{1}{k^2 \log^2 k})\}$. Given $\eta$-approximate query access to an unknown function $f: \R^n \rightarrow \R$ \textcolor{black}{that is bounded in $\ball(\bm 0,2\sqrt{n})$}, there exists a one-sided error tester \testjunta (\autoref{alg:testjuntaapprox}) that performs $\widetilde{O}((k \log k)/\eps)$ queries and guarantees:
\begin{itemize}
    \item \textit{\bf Completeness:} If $f$ is a $k$-junta, then \autoref{alg:testjuntaapprox} always \textsc{Accepts}.

    \item \textit{\bf Soundness:} If $f$ is $\eps$-far from all $k$-juntas, then \autoref{alg:testjuntaapprox} \textsc{Rejects} with probability at least $2/3$.
    
\end{itemize}

\end{theo}

\begin{algorithm}[ht]
\caption{\testjunta }\label{alg:testjuntaapprox}
{\bf Inputs:} $\eta$-Approximate function oracle $\widetilde{f}: \R^n \rightarrow \R$, $k \in \N, \eps, \eta \in (0,1)$.\\

{\bf Output:} Output \textsc{Accept} iff $f$ is a $k$-junta. \\

Choose a random partition $\cB$ of $[n]$ into $r= O(k^2)$ parts. \

Initialize $S \gets [r]$, $I \gets \emptyset$. \

\For{$i=1$ to $O(k/\eps)$}{

$B \gets \mathsf{FindInfBucket}(\widetilde{f}, B, S)$ (\autoref{alg:findinfbucket}).\

\If{$B \neq \emptyset$ ($= B_j$ for some $j \in S$)}
{
    $I \gets I \cup \{j\}$.
    $S \gets S \setminus \{j\}$.
}

\If{$|I|>k$}{
    \Return \textsc{Reject}. \
}\label{line:junta-count-check}

}

\Return \textsc{Accept}. \

\end{algorithm}

In order to prove \autoref{theo:kjuntafinalmain}, we first build some machinery, and then analyze \autoref{alg:testjuntaapprox} and its subroutine (\autoref{alg:findinfbucket} for juntas) in \autoref{sec:k-junta-analysis}. Let us start with the notion of influence.

\begin{defi}[Influence]\label{defn:influence}
Let $f:\R^n \rightarrow \R$ be a function. For any set $S \subseteq [n]$, the influence of $f$ over $S$ with respect to a distribution $\cD$ over $\R^n$ is defined as follows:
\[\infl_f(S)=\E_{\bm x, \bm y \sim \cD} \left[\size{f(\bm y) - f(\bm x_S \bm y_{\overline{S}})}\right].\]
\end{defi}

We now prove a structural result.
Let $\mathcal{J}_k$ denote the set of all $k$-juntas on $n$ variables.

\begin{theo}\label{thm:structure-junta}
If $\dist(f,\mathcal{J}_k)\geq\varepsilon$, then for all $S\subseteq[n]:|S|\leq k$, with $\overline{S}\triangleq [n]\setminus S$, $\infl_f(\overline{S})\geq\varepsilon$.
\end{theo}

\begin{proof}
Fix some $S\subseteq[n]$, such that $|S|\leq k$. Let $\mathcal{J}_S$ denote the set of all juntas on $S$. Note $\mathcal{J}_S\subsetneq\mathcal{J}_k$.
Define $g_S:\R^n\to\R$ as the junta on $S$ that minimizes the distance from $f$:
\[g_S\triangleq\arg\min_{g\in\mathcal{J}_{S}}\{\dist(f,g)\}.\]
    
Observe that $\infl_{g_S}(\overline{S})=0$. We give a method of constructing such a $g_S$, coset-wise.
For every $\bm x\in\R^{|\overline{S}|}$, define a function $f_{\bm x}:\R^n\to\R$ as $f_{\bm x}(\bm y)\triangleq f(\bm x_{\overline{S}} \bm y_S)$. 
Observe that $f_{\bm x}\in\mathcal{J}_S$, and hence $f_{\bm x}\in\mathcal{J}_k$, for every $\bm x\in\R^{|\overline{S}|}$. 
So, from the premise, we have, for every $\bm x\in\R^{|\overline{S}|}$,
\[\dist_{\mathcal{D},\ell_1}(f,f_{\bm x})=\E_{\bm y\sim\mathcal{D}}\left[\size{f(\bm y)- f_{\bm x}(\bm y)}\right]=\E_{\bm y\sim\mathcal{D}}\left[\size{f(\bm y)- f(\bm x_{\overline{S}} \bm y_S)}\right]\geq\varepsilon.\]
\color{black}
Since, the role of $\bm x$ is determined only by the values in the variables in $\overline{S}$, we may as well extend it for all $n$ variables, i.e., for all $\bm x\in\R^n$, we have 
\[\E_{\bm y\sim\mathcal{D}}\left[\size{f(\bm y)- f(\bm x_{\overline{S}} \bm y_S)}\right]\geq \eps.\]
\color{black}
Now we may sample $\bm x$ from any distribution $\mathcal{D}'$ supported on $\R^n$ as well, i.e.,
\[\E_{\bm x\sim\mathcal{D}',\bm y\sim\mathcal{D}}\left[\size{f(\bm y)- f(\bm x_{\overline{S}} \bm y_S)}\right]\geq \eps.\]
In particular, we may set $\mathcal{D}'=\mathcal{D}$ to get
\[\infl_f(\overline{S})=\E_{\bm x,\bm y\sim\mathcal{D}}\left[\size{f(\bm y)- f(\bm x_{\overline{S}} \bm y_S)}\right]\geq \eps.\qedhere\]
\end{proof}

Next we prove a lemma that connects the influence of the union of two sets of variables, with the influences of individual sets.

\begin{lem}[Sub-additivity of Influence]\label{lem:infl-submodular}
For every function $f:\R^n \rightarrow \R$, and any $S, T \subseteq [n]$,
\begin{equation*}
\max\{\infl_f(S),\infl_f(T)\} \leq \infl_f(S \cup T) \leq \infl_f(S) + \infl_f(T).    
\end{equation*}
\end{lem}

\begin{proof}
From \Cref{defn:influence}, we have
\[\infl_f(S)=\E_{\bm x,\bm y\sim\mathcal{D}}\left[\size{f(\bm y)- f(\bm x_{S} \bm y_{\overline{S}})}\right],\; \infl_f(T)=\E_{\bm x,\bm y\sim\mathcal{D}}\left[\size{f(\bm y)- f(\bm x_{T} \bm y_{\overline{T}})}\right]\text{, and }
\]
\begin{eqnarray*}
\infl_f(S\cup T) &=& \E_{\bm x,\bm y\sim\mathcal{D}}\left[\size{f(\bm y)- f(\bm x_{S \cup T} \bm y_{\overline{S \cup T}})}\right] \\
&=& \E_{\bm x,\bm y\sim\mathcal{D}}\left[\size{f(\bm y) - f({\bm x}_S {\bm y}_{\overline{S}}) + f({\bm x}_S {\bm y}_{\overline{S}}) - f(\bm x_{S \cup T} \bm y_{\overline{S \cup T}})}\right] \\
&\leq& \E_{\bm x,\bm y\sim\mathcal{D}}\left[ \size{f(\bm y) - f({\bm x}_S {\bm y}_{\overline{S}})}\right] + \E_{\bm x,\bm y\sim\mathcal{D}} \Big[\Big|f(\underbrace{{\bm x}_S {\bm y}_{\overline{S}}}_{\triangleq\bm y'\sim\cD}) - f(\underbrace{\bm x_{S \cup T} \bm y_{\overline{S \cup T}}}_{=\bm x_T \bm y'_{\overline{T}}})\Big|\Big] \\
&=& \E_{\bm x,\bm y\sim\mathcal{D}}\left[ \size{f(\bm y) - f({\bm x}_S {\bm y}_{\overline{S}})}\right] + \E_{\bm x,\bm y\sim\mathcal{D}}\left[ \size{f(\bm y) - f({\bm x}_T {\bm y}_{\overline{T}})}\right] \\
&=& \infl_f(S) + \infl_f(T).
\end{eqnarray*}
Non-negativity of $\infl$ trivially completes the other part of the lemma.
\end{proof}

Next we show, if $f$ is far from being a $k$-junta, the influence of the complement of $S$ in \autoref{alg:testjuntaapprox} (i.e., all buckets not yet identified as influential) is high, unless we identify more than $k$ parts.

\begin{lem}\label{lem:influence_partition_complement}
Let $f: \R^n \rightarrow \R$ and $\mathcal{B} = \{B_1,\ldots,B_r\}$ be a random partition of $[n]$, for $r = \Theta(k^2)$. 
If $\dist(f,\mathcal{J}_k)\geq\varepsilon$, with probability $\geq 99/100$ over the randomness of the partition, we have $\infl_f(\overline{S}) \geq \eps/4$ for any $S \subseteq [n]$ which is a union of at most $k$ parts of $\mathcal{B}$.
\end{lem}
Before proving it, let us first define the notions of intersecting families, and $p$-biased measure.

\begin{defi}[Intersecting family]\label{label:intersecting-family}
Let $\ell \in\N$. A family of subsets $\cC$ of $[n]$ is \emph{$\ell$-intersecting} if for any two sets $S, T \in \cC$, $\size{S \cap T} \geq \ell$.
\end{defi}

\begin{defi}[$p$-biased measure]\label{def:p-biased-measure}
Let $p \in (0,1)$. Construct a set $S\subseteq[n]$ by including each index $i \in [n]$ in $S$ with probability $p$. Then the \emph{$p$-biased measure} is defined as follows:
\[\mu_p(\cC)= \pr_{S}[S \in \cC].\]
\end{defi}

We will use the following result to bound the $p$-biased measure of intersecting families.

\begin{theo}[\cite{dinur2005hardness,friedgut2008measure}]\label{theo:pmesasure}
Let $\ell \geq 1$ be an integer and let $\cC$ be a $\ell$-intersecting family of subsets of $[n]$. 
For any $p < \frac{1}{\ell+1}$, the $p$-biased measure of $\cC$ is bounded by $\mu_p(\cC) \leq p^{\ell}$.  
\end{theo}

Now we are ready to prove \autoref{lem:influence_partition_complement}. 
We note that the proof is similar to that of {\cite[Lemma 2]{blais2015partially}}. We are adding it here for completeness.

\begin{proof}[Proof of \autoref{lem:influence_partition_complement}]

We will prove: with high probability over the random partition $\cI$, $\infl_f(\overline{S}) \geq \frac{\eps}{4}$, when $S$ is a union of $\leq k$ parts of $\cI$.

Consider the  family of all sets whose complements have influence at most $t \eps$, for some $0 \leq t \leq \frac{1}{2}$:
\[\cF_t= \{S \subseteq [n]: \infl_f(\overline{S}) < t \eps\}.\]

Consider any two sets $S_1, S_2 \in \cF_{1/2}$, i.e., $\max\{\infl_f(\overline{S_1}),\infl_f(\overline{S_2})\}<\eps/2$. 
By \autoref{lem:infl-submodular}, we get
\begin{equation}\label{eq:infl-of-compl-ofint-is-small}
    \infl_f(\overline{S_1 \cap S_2}) = \infl_f(\overline{S_1} \cup \overline{S_2}) \leq \infl_f(\overline{S_1}) + \infl_f(\overline{S_2}) <2\eps/2= \eps.
\end{equation}

As $\dist(f,\mathcal{J}_k)\geq\varepsilon$, for every $S \subseteq [n]$ of size $|S| \leq k$, we have $\infl_f(\overline{S}) \geq \eps$. 
Comparing with \eqref{eq:infl-of-compl-ofint-is-small}, we may infer $|\overline{\overline{S_1} \cup \overline{S_2}}|=|S_1 \cap S_2| > k$.
Since we started with two arbitrary sets $S_1$ and $S_2$ from $\cF_{1/2}$, $\cF_{1/2}$ must then be a $(k+1)$-intersecting family (\autoref{label:intersecting-family}).
Now, consider the two cases:
\begin{itemize}
    \item[(i)]{\bf $\cF_{1/2}$ contains only sets of size at least $2k$.}

We note that $\cF_{1/4}$ is a $2k$-intersecting family. To see this, let us assume that it is not the case. Then there exist $T_1,T_2\in\cF_{1/4}:|T_1\cap T_2|< 2k$. So, by \autoref{lem:infl-submodular}, we get
\begin{equation*}
    \infl_f(\overline{T_1 \cap T_2}) = \infl_f(\overline{T_1} \cup \overline{T_2}) \leq \underbrace{\infl_f(\overline{T_1})}_{<\eps/4} + \underbrace{\infl_f(\overline{T_2})}_{<\eps/4} <2\eps/4< \eps/2.
\end{equation*}
This means that $T_1 \cap T_2 \in \cF_{1/2}$ (by definition of $\cF_{1/2}$) with $|T_1 \cap T_2| < 2k$, which contradicts our assumption, in this case, that all sets in $\cF_{1/2}$ have cardinality $\geq 2k$.

Let $S \subseteq [n]$ be a union of $k$ parts in $\cI$. 
Since $\cI$ is a random $r$-partition of $[n]$, $S$ is a random subset obtained by including each element of $[n]$ in $S$ with probability $\frac{k}{r} \leq \frac{1}{2k+1}$. By \Cref{theo:pmesasure}, we can bound the following:
\[\pr\left[\infl_f(\overline{S}) < \frac{\eps}{4}\right] = \pr[S \in \cF_{1/4}]= \mu_{k/r}\left(\cF_{1/4}\right) \leq \left(\frac{k}{r}\right)^{2k}.\]

Now we apply the union bound over all possible choices of such $S$ (a union of $k$ parts of $\cI$). 
Then, the probability that at least one such $S$ has $\infl_f(\overline{S}) < \eps/4$ is at most
\[{r \choose k }\left(\frac{k}{r}\right)^{2k} \leq \left(\frac{er}{k}\right)^k\left(\frac{k}{r}\right)^{2k} \leq \left(\frac{ek}{r}\right)^k = O(k^{-k}) \ll \frac{1}{100}.\]

\item[(ii)]{\bf $\cF_{1/2}$ contains at least one set of size less than $2k$}

Let $J \in \cF_{1/2}$ be such that $|J| < 2k$. 
For any random partition $\cI$ with $\Theta(k^2)$ parts, all the elements of $J$ are in different parts of $\cI$, with high probability ($\geq 99/100$). 
For any $T \in \cF_{1/2}$, $|J \cap T| \geq k+1$ (since, $\cF_{1/2}$ is a $(k+1)$-intersecting family), and hence $T$ is not covered by any union of $k$ parts of $\cI$. That is, for any $S \subseteq [n]$ which is covered by $\leq k$ parts of $\cI$, $S \not\in \cF_{1/2}$ which implies $\infl_f(\overline{S}) \geq \eps/2$. \qedhere
\end{itemize}
\end{proof}

\subsection{Proof of correctness of \texorpdfstring{\testjunta}{test-k-junta}}\label{sec:k-junta-analysis}

As \autoref{alg:testjuntaapprox} invokes \autoref{alg:findinfbucket}, we first prove its correctness for juntas.

\begin{defi}[Influential variables and buckets]
Let $f: \R^n \rightarrow \R$, and $k \in \N$. A variable ${\bm x}_i,i \in [n]$ is said to be \emph{influential} with respect to $f$, if for some $\bm x\in\R^n$, only changing ${\bm x}_i$ changes the value of $f(\bm x)$. 
A bucket $B \subseteq [n]$ is said to be influential if it contains the index of at least one influential variable, and non-influential otherwise.
\end{defi}
Let $\mathcal{E}$ denote the event that for $\bm x,\bm y\sim\cN(\bm 0,I_n)$, $\bm x,\bm y\in\ball(\bm 0,2\sqrt{n})$. 
By \autoref{fact:gaussian-concentration}, we have $\Pr_{\bm x,\bm y\sim\cN(\bm 0,I_n)}[\overline{\mathcal{E}}]=\Pr_{\bm x,\bm y\sim\cN(\bm 0,I_n)}[\bm x,\bm y\not\in\ball(\bm 0,2\sqrt{n})\leq 0.01$. 
Let \[\kappa\triangleq \minop_{\substack{V \subseteq [n]\\V \mbox{ is influential}}} \left\{\E_{\bm x, \bm y\sim\cN(\bm 0,I_n)} \left[|f({\bm x}_V {\bm y}_{\overline{V}}) - f(\bm y)|\mid\mathcal{E}\right]\right\}.\]

As mentioned before, we use $\findbucket$ (\autoref{alg:findinfbucket}) for testing $k$-linearity as well. 
However, the proof of correctness of \autoref{alg:findinfbucket} in \autoref{sec:klinearitytest}, presented in \autoref{claim:close-to-linear-implies-good_new_bucket}, holds only for linear functions, which may not be the case in general. 
So, we prove the same for juntas as well.

\begin{cl}[Correctness of \autoref{alg:findinfbucket} for juntas]\label{claim:close-to-linear-implies-good_new_bucket_junta}
Let $f: \R^n \rightarrow \R$ be given via the $\eta$-approximate query oracle $\tilde f$, where $\eta \leq \min\{\kappa/4, 1/(1000k^2 \log^2k)\}$, $B
 = \{B_1,\ldots,B_r\}$ be a partition of $[n]$, and $\emptyset\neq S \subseteq [r]$. With $\kappa\geq\max\{ 2\|f\|_{\infty,\ball(\bm 0,2\sqrt{n})},\eps/(4k^2)\}$, $\mathsf{FindInfBucket}(f,B,S)$ guarantees:
    \begin{enumerate}
        \item If none of the $B_i$'s are influential, $\mathsf{FindInfBucket}(f, B, S)$ always returns $\emptyset$ and performs exactly $2$ queries to $f$.
        
        \item Otherwise, with probability at least $1 - 16\eta\lceil\log |S|\rceil^2/\kappa$, $\mathsf{FindInfBucket}(f,B,S)$ returns $B_j$ for some $j \in S$ which is $\kappa$-influential and performs $\leq 8 \ceil{\lg(|S|)}^2$ queries to $f$.
    \end{enumerate}
\end{cl}

\begin{proof}
Let $\vec{x}$ and $\vec{y}$ be the Gaussian random vectors sampled by \autoref{alg:findinfbucket}, and $\vec{w} \triangleq  \vec{y}_{\overline{V}} \vec{x}_V\in \R^n$.

\begin{enumerate}
    \item As none of the $B_i$'s are influential, $V (= \{i \in [n] : i \in B_j \mbox{ for some } j \in S\})$ does not contain any influential variable. So, $f \equiv f\mid_{\overline{V}}$, where $f\mid_{\overline{V}}: \R^{[|n] \setminus V|} \rightarrow \R$, i.e., for every $\bm x\in\R^{n}$, with $\bm x_{\overline{V}}\triangleq (\bm x_i,i\not\in V)\in\R^{[|n] \setminus V|}$, $f(\bm x)=f\mid_{\overline{V}}({\bm x_{\overline{V}}})$. 
    So, in Line $5$ of \autoref{alg:findinfbucket}, $f(\bm w)=f(\bm y)$, and hence 
    $|\widetilde f(\bm w)-\widetilde f(\bm y)|=|\widetilde f(\bm w)-f(\bm w)-(\widetilde f(\bm y)-f(\bm y))|
\leq \overbrace{|\widetilde f(\bm w)-f(\bm w)|}^{\leq\eta}+\overbrace{|\widetilde f(\bm y)-f(\bm y)|}^{\leq\eta}\leq 2\eta$, 
    by the triangle inequality, implying it will always return $\emptyset$, performing exactly $2$ queries to $f$.

\item

We will show: When $S$ has some $j$ such that $B_j$ contains a $\kappa$-influential variable, then the test in lines 5-10 will fail with probability $1-8\eta\lceil\log |S|\rceil/\kappa$. Let $Z \triangleq |f({\bm x}_V {\bm y}_{\overline{V}}) - f({\bm y})|$.

If $\E_{\bm x, \bm y} \left[Z\right] \geq \kappa$, 
then, we argue (by anticoncentration) that $\Pr_{\bm x, \bm y}\left[Z > 4\eta\right] \geq 1 - 8\eta\lceil\log |S|\rceil/\kappa$, and hence $\Pr_{\bm x, \bm y}\left[|\widetilde{f}({\bm x}_V {\bm y}_{\overline{V}}) - \widetilde{f}(\bm y)| > 2\eta\right]$ with the same probability. 
To use anticoncentration on $Z$, we first bound its second moment, conditioned on the event $\mathcal{E}$: 
\[\E_{\bm x,\bm y\sim\cN(\bm 0,I_n)}[Z^2\mid\mathcal E]= \E_{\bm x,\bm y\sim\cN(\bm 0,I_n)}[f({\bm x}_V {\bm y}_{\overline{V}})^2 + f({\bm y})^2-2f({\bm x}_V {\bm y}_{\overline{V}})f({\bm y})\mid\mathcal{E}]\leq 4\|f\|^2_{\infty,\ball(\bm 0,2\sqrt{n})}.\]

Then, applying \autoref{thm:paley-zygmund} (with $\theta=\frac{4\eta}{\kappa}\in(0,1)$), conditioned on $\mathcal{E}$, we get
\[\Pr[Z \geq 4\eta\mid\mathcal E] \geq \Pr\Bigg[Z\geq \frac{4\eta}{\kappa}\underbrace{\E[Z]}_{\geq \kappa}\mid\mathcal{E}\Bigg]\geq \left(\! 1-\frac{4\eta}{\kappa}\right)^2 \frac{\E^2[Z\mid\mathcal{E}]}{\E[Z^2\mid\mathcal{E}]} \geq \left(\! 1-\frac{4\eta}{\kappa}\right)^2 \frac{\kappa^2}{4\|f\|^2_{\infty,\ball(\bm 0,2\sqrt{n})}}.\]

From the premise, we have $\kappa\geq 2\|f\|_{\infty,\ball(\bm 0,2\sqrt{n})}$, giving us $\Pr[Z \geq 4\eta\mid\mathcal E]\geq (1-4\eta/\kappa)^2$.

Thus, with probability at least $1-8\eta/\kappa$, the condition in Line 5 will not hold, ensuring \autoref{alg:findinfbucket} reaches the recursive step (Lines 8--13). By construction, $S = S^{(L)} \sqcup S^{(R)}$, with $1 \leq |S^{(L)}| \leq \ceil{|S|/2}$ and $1 \leq |S^{(R)}| \leq \floor{|S|/2}$. We have two possibilities:
    \begin{enumerate}[label=(\roman*)]
        \item $S^{(L)}$ contains no influential buckets and there exists $\emptyset\neq S_R \subseteq S^{(R)}$ such that $B_{j}$ is influential for all $j \in S_R$.
        
        \item There exists $\emptyset\neq S_L \subseteq S^{(L)}$ such that $B_{j}$ is influential for all $j \in S_L$.
    \end{enumerate}

    In Case (i), $\mathsf{FindInfBucket}(f,B,S^{(L)})$ will always return $\emptyset$ and perform $2$ queries (by Part 1) and hence the return value of $\mathsf{FindInfBucket}(f,B,S)$ will be $\mathsf{FindInfBucket}(f,B,S^{(R)})$. Using the strong induction hypothesis we get, with probability at least $1-8\eta\lceil\log |S|\rceil/\kappa$, this return value will be $\{j\}$ for some $j \in S_R$ and the total number of queries made will be $\leq 2 + 2 + 4 \ceil{\lg(\floor{|S|/2})} \leq 4 + 4(\ceil{\lg |S|} - 1) = 4 \ceil{\lg |S|}$.

    In Case (ii), again using the strong induction hypothesis, with probability at least $1-8\eta\lceil\log |S|\rceil/\kappa$, $\mathsf{FindInfBucket}(f,B,S^{(L)})$ will return $\{j\}$ for some $j \in S_L$ and thus the algorithm will return $\{j\}$ (line 13). The number of queries made will be $\leq 2 + 4 \ceil{\lg(\ceil{|S|/2})} \leq 2 + 4(\ceil{\lg |S|}-\frac{1}{2}) = 4 \ceil{\lg |S|}$.
    \end{enumerate}

    Note that, if we run $\mathsf{FindInfBucket}(f, B, S)$ for some set $S$ of bucket-indices with $|S| > 1$, each recursive step (irrespective of whether it falls in case (i) or case (ii) above) will succeed with probability $\geq 1-8\eta\lceil\log |S|\rceil/\kappa$. 
    For the top level call to succeed, all the recursive calls must succeed. 
    The number of recursive calls made is at most $2 \lceil \lg |S| \rceil$. 
    So, we do a union bound over all the failure events to bound the failure probability of the top-level call by $16 \eta\lceil\log |S|\rceil^2/\kappa$.
\end{proof}

We are now ready to prove the main theorem of this section:

\begin{proof}[Proof of \texorpdfstring{\Cref{theo:kjuntafinalmain}}{thm-k-junta}]
Let us start with the completeness proof.

\textbf{Completeness:}
The argument of completeness of \autoref{alg:testjuntaapprox} holds easily, as when $f$ is a $k$-junta, at most $k$ of the buckets in the random partition $\mathcal{B}$ will be influential. From \autoref{claim:close-to-linear-implies-good_new_bucket_junta}, we get that each execution of the $\mathsf{For}$ loop increments $I$ by: either the index of an influential bucket (if it finds one), or nothing, while also removing the bucket from future loop executions. Hence after all its executions, $|I|$ is incremented at most $k$ times, falsifying the check in \autoref{line:junta-count-check}, making $f$ $\textsc{Accept}$.

\textbf{Soundness:}
    Now we prove soundness. 
    Here our goal is to prove that if $f$ is $\eps$-far from all $k$-juntas, then \autoref{alg:testjuntaapprox} rejects it with probability at least $2/3$. 
    In particular, if $f$ is $\eps$-far from being $k$-junta, and $S$ is the union of at most $k$-parts of $\cI$, then Line 7 of \autoref{alg:testjuntaapprox} will be satisfied with high probability.
    
    By \autoref{lem:influence_partition_complement}, we have: if $|I|\leq k$, $\infl_f(S) \geq \eps/4$ ($S=[n]\setminus I$ in the execution of \autoref{alg:testjuntaapprox}). 
    At a point in the execution, when exactly $k$ influential buckets have been identified, $\infl_f(S) \geq \eps/4$, $|S|=O(k^2)-k=O(k^2)$, and thus, by \autoref{lem:infl-submodular}, there must exist a bucket $B_i,i\in S$ such that $\infl_f(B_i)\geq\eps/(4k^2)$.
    This bucket will then be output by $\mathsf{FindInfBucket}(f,B,S)$.
    So, until $k + 1$ influential buckets have been identified, each invocation of $\mathsf{FindInfBucket}(f,B,S)$ in an execution of its $\mathsf{For}$ loop, returns an influential bucket with probability at least $1-64\eta\lceil\log k\rceil^2/\kappa$. 
    Treating the outcome of each invocation as a geometric random variable, we conclude: to recover at least $k+1$ such buckets with probability $\geq 2/3$, \[ k/\left(1-64\eta\lceil\log k\rceil^2/\kappa\right)=k+k\sum_{i=1}^\infty (64\eta\lceil\log k\rceil^2/\kappa)^i=O(k/\eps)\] iterations of the $\mathsf{For}$ loop suffice. When $|I| > k$, in \autoref{line:junta-count-check}, algorithm rejects with the same probability.
    
    \textbf{Query Complexity:} From \autoref{claim:close-to-linear-implies-good_new_bucket_junta}, we know: each invocation of $\mathsf{FindInfBucket}(f,B,S)$ performs $\leq 4\log|S|$ queries to $f$. 
    With $O(k/\eps)$ such invocations in total, and $|S|=r=O(k^2)$, the overall query complexity thus is $O(k\log k/\eps)$.
\end{proof}

\section{Lower Bounds}\label{sec:lowerbound}

In this section, we will present our lower bound results, restated here for convenience:

\lowerbounds*
For proving the lower bounds, we follow the approach of \cite{blais2012property}, i.e., we show the hardness of testing these problems by reducing the \textsc{Set-Disjointness} problem, a canonical hard problem in the field of communication complexity, to them. 
Proving lower bounds in the exact query model directly gives us lower bounds in the approximate query model.
For brevity, we only present the lower bound argument for $k$-linearity. The other results follow the same approach.

\textbf{A lower bound of $\Omega(\frac{1}{\eps})$ queries:}
For any $\eps \in (0,1)$, it is a folklore result that distinguishing if a function $f$ is linear, or $\eps$-far from being linear in $\ell_1$-distance, requires $\Omega(1/\eps)$ queries. An exposition for this is provided in \cite{fischer2024basiclowerboundproperty}.

\textbf{A lower bound of $\Omega(d)$ queries for $k$-sparse, degree-$d$ polynomial testing:}
As before, it is another folklore result that given a proximity parameter $\eps \in (0,1)$, $\Omega(d)$ queries are necessary to distinguish if an unknown function $f$ is a degree-$d$ polynomial, or a degree-$(d+1)$ (or $(d-1)$) polynomial (and hence, is $\eps$-far from all degree-$d$ polynomials) 
with probability at least $2/3$. 

To prove the $\Omega(k)$ lower bound, we next present a reduction from \textsc{Set-Disjointness}.

\subsection{Communication complexity setting}

We consider the two-player communication game setting, featuring Alice and Bob. 
Alice has a function $f$, and Bob has a function $g$, and they jointly want to evaluate another function $h$ (which is a function of $f$ and $g$). 
We assume that both the players have unbounded computational power, but to evaluate $h$, they need to communicate among themselves, and the objective is to keep this communication overhead as small as possible. 

It is known that it takes linear in the size of the bits of the largest set to solve \textsc{Set-Disjointness}.

\begin{theo}[\cite{haastad2007randomized}, \textsc{Set-Disjointness} Lower Bound]\label{theo:setdisjointnesslb}
Let Alice and Bob have two sets $A$ and $B$, respectively, each of size at most $k$ from a universe of size $n$. In order to distinguish if $|A \cap B|=1$, or $|A \cap B|=0$, $\Omega(k)$ bits of communication between Alice and Bob are required.
\end{theo}

\subsection{Connection between Communication and Query complexity}

Consider two functions $f$ and $g$, a property $\cP$, $h\triangleq f-g$, and a communication problem $C_{h,\cP}$: 
Alice and Bob receive $f$ and $g$, respectively, and they want to decide if $h\in\cP$, or $f$ is $\eps$-far from $\cP$. 
\cite{blais2012property} proved that $\cR(C_{h,\cP})$, the randomized communication complexity of $C_{h,\cP}$, is at most twice $Q(\cP)$, the query complexity of deciding the property $\cP$.

\begin{lem}[{\cite[Lemma 2.2]{blais2012property}}]\label{lem:querycommunicationbound}
Let $\cP$ be a property of functions, and $h$ be a function. Then $\cR(C_{h,\cP}) \leq 2 Q(\cP)$.
\end{lem}

Now, given two sets $A$ and $B$, we construct two functions $f$ and $g$ suitably as follows.

\subsection{Construction of the hard instances for the lower bound}
Given $A \subseteq [n]:|A|=k$, Alice constructs a polynomial $f= \sum_{i \in A} {\bm x}_i$. 
Similarly, Bob constructs a polynomial $g = \sum_{i \in B} {\bm x}_i$. Note that $f,g:\R^n\to\R$. 
Let $h=f-g$.
Consider the two cases:

\begin{itemize}
    \item[(i)]{\bf $|A \cap B|=1$:} In this case $h$ would be a $(2k+2)$ linear function.

    \item[(ii)]{\bf $|A \cap B|=0$:} In this case, $h$ would be a $2k$-linear function.
    
\end{itemize}

Now we show that under the standard Gaussian distribution $\cN(\bm 0,I_n)$, $h$'s corresponding to cases (i), and (ii) above, are sufficiently far in $\ell_1$-distance with probability at least $2/3$.

\begin{lem}\label{lem:lbfar}
Let $f_1$ be a $(2k+2)$-linear function, and $f_2$ be a $2k$-linear function with coefficients from $\{0,1\}$. Under the standard Gaussian distribution $\cN(\bm 0,I_n)$, $f_1$ is $\Omega(1)$-far from $f_2$ in $\ell_1$-distance with probability at least $2/3$.  
\end{lem}

\begin{proof}
Consider the function $g=f_1-f_2$. Observe that $g$ is at least a $2$-linear function, with at least the two terms from $f_1$ not present in $f_2$. Without loss of generality, let $g(\bm x)= \bm x_1+ \bm x_2$.
We are interested in the event where
$|g(\bm x)|>O(1)$.

Using \Cref{srt_thm:glazer-mikulincer} (with $d=1,t=0$, $g({\bm x}) = {\bm x_1} + {\bm x_2}\text{, and coeff}_1(g)=\sqrt{2}$), we get: $\forall\eps>0$,

\[\Pr_{\vec{x} \sim \cN(\bm 0,I_n)}\Big[|g(\vec{x})| \leq \eps \Big] \leq C\left(\frac{\eps}{\sqrt{2}}\right).\]

Setting $\eps= \frac{1}{2C}$, we have the following:
\[\Pr_{\vec{x} \sim \cN(\bm 0,I_n)}\left[|g(\vec{x})| \leq \frac{1}{2C}\right] \leq \left(\frac{1}{2\sqrt{2}}\right) \leq 1/3.\]

Thus, $f_1$ is at least $\frac{1}{2C}$-far from $f_2$ with probability at least $2/3$. 
\end{proof}

\subsection*{Simulation}
Let $C_h$ be the communication problem where Alice and Bob have two sets $A$ and B, respectively, each of size $k$. 
As discussed before, they have constructed the two functions $f$ and $g$, respectively, from $A$ and $B$. 
Equivalently, we can also say that Alice and Bob are given two functions $f$ and $g$ respectively, with the promise that the function $h=f-g$ is either a $2k$-linear, or a $(2k+2)$-linear function, and they should accept if and only if $h$ is a $2k$-linear function.

From the above reduction, it is clear that $\cR(C_h) \geq \cR(\textsc{Set Disjointness}) = \Omega(k)$. By \Cref{lem:querycommunicationbound}, we can say that any testing algorithm that can distinguish between $2k$-linear and $(2k+2)$-linear functions with probability at least $2/3$, requires $\Omega(k)$ queries. Moreover, from \autoref{lem:lbfar}, we know that any $2k$-linear function is $\Omega(1)$-far from any $(2k+2)$-linear function. Thus, any $2k$-linearity tester (with $\eps = o(1)$) that uses $o(k)$ queries can not distinguish $2k$-linear functions (which should be accepted) and $(2k+2)$-linear functions (which should be rejected) correctly, with probability $\geq 2/3$. 
This completes the proof of the lower bound of $k$-linearity testing.

\bibliographystyle{alpha}
\bibliography{reference}

\appendix

\section{Omitted Algorithms from the Main Body}

Here we state the algorithms from \cite{fleming2020distribution}, \cite{arora2} and \cite{arora2024optimaltestinglinearitysosa} for completeness. 

We begin with \autoref{alg:zero-mean-additivity}, a query-optimal approximate additivity tester, with its subroutines provided in \autoref{alg:subroutines_noisy}. The properties of this tester were recorded earlier in \autoref{srt_thm:additivity-tester-arora2}. This is followed by \autoref{alg:low_degree_main_algorithm_approx}, a query-optimal (w.r.t. the distance parameter $\eps$) approximate low-degree tester, with its subroutines in \autoref{alg:subroutines_approx}, and its properties recorded in \autoref{thm:approximate_low_degree_tester}. 

\begin{algorithm}[ht]
  \caption{{\cite[Algorithm 7]{arora2} and \cite[Algorithm 5]{arora2024optimaltestinglinearitysosa}:} {(Query-)Optimal Approximate Additivity Tester}}\label{alg:zero-mean-additivity}
  \Procedure{\textsc{Approximate Additivity Tester}($f,\mathcal{D},\alpha,\eps,R$)}
  {
    \Given{Query access to $f\colon \mathbb{R}^n \to \mathbb{R}$, sampling access to an unknown $(\eps/4,R)$-concentrated distribution $\mathcal{D}$, a noise parameter $\alpha>0$, and a farness parameter $\eps> 0$;}
    $\delta\gets 3\alpha$, $r\gets 1/50$\;
\Return  \textsc{Reject} if \textsc{TestAdditivity}($f,\delta$) returns \textsc{Reject}\;
    \For{$N_{\ref{alg:zero-mean-additivity}} \gets O(1/\varepsilon)$ times}{
      Sample $\bm p \sim \mathcal{D}$\;
    \If{$\bm p\in\ball(\bm 0,\lrad)$}{
      \Return \textsc{Reject} if $|f(\bm p) - $ \textsc{Approximate-}$g$($\bm p,f,\delta$)$|>5\delta n^{1.5}\kappa_{\bm p}$, or if \textsc{Approximate-}$g$($\bm p,f,\delta$) returns \textsc{Reject}.
    }}
    
\Return    \textsc{Accept}.}
\end{algorithm}

\begin{algorithm}[ht]
  \caption{{\cite[Algorithm 8]{arora2} and \cite[Algorithm 6]{arora2024optimaltestinglinearitysosa}} Additivity Subroutines} \label{alg:subroutines_noisy}
  \Procedure{\textsc{TestAdditivity}($f,\delta$)}
  {
    \Given{Query access to $f\colon \mathbb{R}^n \to \mathbb{R}$, threshold parameter $\delta>0$;}
    \For{$N_{\ref{alg:subroutines_noisy}} \gets O(1)$ times}{
      Sample $\bm{x},\bm{y},\bm{z} \sim \cN(\bm 0, I_n)$\;
      \Return \textsc{Reject} if $|f(-\bm x)+f(\bm x)|>\delta$\;
      
      \Return \textsc{Reject} if $|f(\bm x - \bm y) - \left(f(\bm x) - f(\bm y)\right)|>\delta$\;
      
      \Return \textsc{Reject} if $\left|f\left(\frac{\bm x - \bm y}{\sqrt{2}} \right) - \left( f \left(\frac{\bm x -\bm z}{\sqrt{2}} \right) + f \left(\frac{\bm z - \bm y}{\sqrt 2} \right)\right)\right|>\delta$\; 
    }
\Return    \textsc{Accept}.
  }
  \Procedure{\textsc{Approximate}-$g$($\bm p, f,\delta$)}
  {
    \Given{$\bm p \in \mathbb{R}^n$, query access to $f\colon \mathbb{R}^n \to \mathbb{R}$, threshold parameter $\delta>0$;}
    Sample $\bm x_{1}
    \sim \mathcal{N}(\bm 0,I_n)$\;
    \Return $\kappa_{\bm p} \left( f( \bm p/\kappa_{\bm p} - \bm x_1) + f(\bm x_1) \right)$.
  }
\end{algorithm}

\begin{algorithm}[ht]
\caption{{\cite[Algorithm 3]{arora2} and \cite[Algorithm 7]{arora2024optimaltestinglinearitysosa}} Optimal Approximate Low Degree Tester} \label{alg:low_degree_main_algorithm_approx}
\Procedure{\textsc{ApproxLowDegreeTester}($f,d,\mathcal{D},\alpha,\eps,\lrad,L$)}
{
  \Given {Query access to $f\colon \mathbb{R}^n \to \mathbb{R}$ that is bounded in $\ball(\bm 0,L)$ for some $L>0$, a degree $d\in\mathbb{N}$, sampling access to an unknown $(\varepsilon/4, R)$-concentrated distribution $\mathcal{D}$, a noise parameter $\alpha > 0$, and a farness parameter $\varepsilon> 0$. 
  }
   $\delta\gets 2^{d+1}\alpha$, $r\gets (4d)^{-6}$\;

\Return \textsc{Reject} if \textsc{ApproxCharacterizationTest}  \textbf{rejects}\;
   \For{$N_{\ref{alg:low_degree_main_algorithm_approx}} \gets O(\varepsilon^{-1})$ times}{
      Sample $\bm{p} \sim \frac{2d\sqrt{n}}{L}\mathcal{D}$;\\
      \If{$\bm p\in\ball(\bm 0,2d\lrad\sqrt{n}/L)$}{
     \Return  \textsc{Reject} if $|f(\bm{p}) -$ \textsc{ApproxQuery-}$g$($\bm{p}$)$|>2\cdot 2^{(2n)^{45d}}(\lrad/L)^d\delta$, or if \textsc{ApproxQuery-}$g$($\bm{p}$)  \textbf{rejects}.}
    }
\Return    \textsc{Accept}.
  }  
  
\end{algorithm}

\begin{algorithm}[!t]
\caption{{\cite[Algorithm 4]{arora2} and \cite[Algorithm 8]{arora2024optimaltestinglinearitysosa}:} Approximate Subroutines} \label{alg:subroutines_approx}
[Recall $\alpha_{i}\triangleq (-1)^{i+1}\binom{d+1}{i}$ and $\delta=2^{d+1}\alpha$.]\\
\Procedure{\textsc{ApproxCharacterizationTest}}
{
$N_{\ref{alg:subroutines_approx}} \gets O(d^2)$ \;
\For{$N_{\ref{alg:subroutines_approx}}$ times }{
\For{$j\in \{1,\dots, d+1\}$}{
    \For{$t\in \{0,\dots, d+1\}$}{
      Sample $\bm{p}\sim\mathcal{N}(\bm{0},j^2(t^2+1) I_n),\bm{q}\sim\mathcal{N}(\bm{0}, I_n)$;\\
      \Return \textsc{Reject} if $|\sum_{i=0}^{d+1}\alpha_i\cdot f(\bm{p}+i\bm{q})|>\delta$\;
      Sample $\bm{p}\sim\mathcal{N}(\bm{0},j^2 I_n),\bm{q}\sim\mathcal{N}(\bm{0},(t^2+1) I_n)$;\\
      \Return \textsc{Reject} if $|\sum_{i=0}^{d+1}\alpha_i\cdot f(\bm{p}+i\bm{q})|>\delta$\;
      }
      Sample $\bm{p},\bm{q}\sim\mathcal{N}(\bm{0},j^2 I_n)$;\\
      \Return \textsc{Reject} if $|\sum_{i=0}^{d+1}\alpha_i\cdot f(\bm{p}+i\bm{q})|>\delta$\;
    }}
\Return     \textsc{Accept}\;
    }

\Procedure{\textsc{ApproxQuery-}$g$($\bm{p}$)}
{

\If{$\bm{p}\in \ball(\bm{0},\srad)$}
{\textbf{return} \textsc{ApproxQuery-$g$-InBall}($\bm{p}$)\;}

\For{$i\in \{0,1,\hdots,d\}$}{
 $c_i\gets \frac{\srad}{\Vert \bm{p}\Vert_{2}} \cos\big(\frac{\pi(i+1/2)}{d+1}\big)$\;
 $v({c_i})\gets$ \textsc{ApproxQuery-$g$-InBall}($c_i\bm p$)  \;
}
Let $p_{\bm p} \colon \mathbb{R}\to\mathbb{R}$ be the unique degree-$d$ polynomial such that $p_{\bm p}(c_i) = v(c_i)$  for all $i$\;
\textbf{return} $p_{\bm p}(1)$\;
}

\Procedure{\textsc{ApproxQuery-$g$-InBall}($\bm{p}$)}
{
Sample $\bm q_1
\sim \mathcal{N}(\bm 0, I_n)$\;

\textbf{return} $ \sum_{i=1}^{d+1} \alpha_i \cdot f(\bm{p}+i\bm{q}_1)$\;
}
\end{algorithm}

\section{Omitted Proofs from \texorpdfstring{\Cref{sec:ksparsepolynomialtest}}{Section 5}}\label{sec:omitted-proofs}

\generalizedhenkelmatrix*
\begin{proof}
	Since for all $i \in [k]$ and $\alpha,\beta \in \N_{\geq 0}$, $M_i(\vec{x}^\alpha) M_i(\vec{x}^\beta)= M_i(\vec{x}^{\alpha+\beta})$, we get: for all $\vec{u} \in \R^n$,
	
 \begin{align*}
		& H_{\ell+1}(f,\bm u) = 
		\begin{pmatrix}
			\sum_{i=1}^k a_i M_i(\bm u^0) & \sum_{i=1}^k a_i M_i(\bm u^1) & \hdots & \sum_{i=1}^k a_i M_i(\bm u^{\ell}) \\
			\sum_{i=1}^k a_i M_i(\bm u^1) & \sum_{i=1}^k a_i M_i(\bm u^2) & \hdots & \sum_{i=1}^k a_i M_i(\bm u^{\ell+1}) \\
			\vdots & \vdots & \ddots & \vdots \\
			\sum_{i=1}^k a_i M_i(\bm u^{\ell}) & \sum_{i=1}^k a_i M_i(\bm u^{\ell+1}) & \hdots & \sum_{i=1}^k a_i M_i(\bm u^{2\ell})
		\end{pmatrix} \\
		& = \begin{pmatrix}
			M_1(\bm u^0) & M_2(\bm u^0) & \hdots & M_k(\bm u^0) \\
			M_1(\bm u^1) & M_2(\bm u^1) & \hdots & M_k(\bm u^1) \\
			M_1(\bm u^2) & M_2(\bm u^2) & \hdots & M_k(\bm u^2) \\
			\vdots & \vdots & \ddots & \vdots \\
			M_1(\bm u^{\ell}) & M_2(\bm u^{\ell}) & \hdots & M_k(\bm u^{\ell})
		\end{pmatrix}
		\begin{pmatrix}
			a_1 M_1(\bm u^0) & a_1 M_1(\bm u^1) & \hdots & a_1 M_1(\bm u^{\ell}) \\
			a_2 M_2(\bm u^0) & a_2 M_2(\bm u^1) & \hdots & a_2 M_2(\bm u^{\ell}) \\
			\vdots & \vdots & \ddots & \vdots \\
			a_k M_k(\bm u^0) & a_k M_k(\bm u^1) & \hdots & a_k M_k(\bm u^{\ell})
		\end{pmatrix} \\
		& = \underbrace{\begin{pmatrix}
				1 & 1 & \hdots & 1 \\
				M_1(\bm u) & M_2(\bm u) & \hdots & M_k(\bm u) \\
				(M_1(\bm u))^2 & (M_2(\bm u))^2 & \hdots & (M_k(\bm u))^2 \\
				\vdots & \vdots & \ddots & \vdots \\
				(M_1(\bm u))^{\ell} & (M_2(\bm u))^{\ell} & \hdots & (M_k(\bm u))^{\ell}
		\end{pmatrix}}_{h_1:\, \R^k \rightarrow \R^{\ell+1}}
		\underbrace{\begin{pmatrix}
			a_1 & 0 & \hdots & 0 \\
			0 & a_2 & \hdots & 0 \\
			\vdots & \vdots & \ddots & \vdots \\
			0 & 0 & \hdots & a_k
		\end{pmatrix}}_{h_2:\, \R^k \rightarrow \R^{k}}
		\underbrace{\begin{pmatrix}
				1 & M_1(\bm u) & \hdots & (M_1(\bm u))^{\ell} \\
				1 & M_2(\bm u) & \hdots & (M_2(\bm u))^{\ell} \\
				\vdots & \vdots & \ddots & \vdots \\
				1 & M_k(\bm u) & \hdots & (M_k(\bm u))^{\ell}
		\end{pmatrix}}_{h_3:\, \R^{\ell+1} \rightarrow \R^k}.
	\end{align*}
	Thus, if $\ell+1>k$, the mapping $h_3$ in the above expression is necessarily non-injective, and hence $H_{\ell+1}(f,\bm u)$ is necessarily singular, for any $\bm u$. This implies that $\det\paren{H_{\ell+1}(f,\bm x)}\equiv 0$.
	If $k \geq \ell+1$, the determinant expansion in the observation follows from the following argument in \cite{bot88}:
	\begin{enumerate}
		\item For any fixed $\bm u$, by standard determinant expansion, $\det (H_{\ell+1}(f, \bm u))$ is a polynomial $Q(\bm a)$, and the total degree of each monomial (in the ``variables'' $a_1,\ldots,a_k$) in $Q$ is exactly $\ell + 1$.
  
		\item If $\mathrm{rank}(H_{\ell+1}(f,\bm u)) \leq \|[a_1 \cdots a_k]\|_0 < \ell + 1$, we must have $\det\paren{H_{\ell+1}(f,\bm u)} \equiv 0$, irrespective of the values of the non-zero $\{a_1,\ldots,a_k\}$, for any $\bm u\in\R^n$. Hence each monomial (in $a_1,\ldots,a_k$) of $Q(a_1,\ldots,a_k)$ must have at least $\ell+1$ of the variables $\{a_i\}$. Since its total degree is exactly $\ell + 1$, each monomial must be of the form $\prod_{i \in S} a_i$ for some $S \subseteq [k]\text{, with } |S| = \ell + 1$.
  
		\item The coefficient of monomial $\prod_{i \in S} a_i$ in $Q$ will be $Q(c_1,\ldots,c_k)$ with $c_i = \indic\{i \in S\}$. But, by the decomposition above, this will be the square of the Vandermonde determinant (see \Cref{srt_defi:vandermonde}), $\det\paren{V_{\ell+1}\left(\{x_i\}_{i \in S}\right)} = \prod_{\substack{i, j \in S\\i < j}} (M_j(\bm u) - M_i(\bm u))$.
	\end{enumerate}
Summing these terms over all $S \subseteq[k]$ with $|S|=\ell+1$, completes the determinant expansion.
\end{proof}

\approxhankelstructure*
\begin{proof}[Proof of \Cref{obs:approx-hankel-structure}]
Let $\vec{u} \in \R^n$, and $\alpha_i \triangleq \tilde{f}(\vec{u}^i) - f(\vec{u}^i)$ for each $i \in\{ 0,1,\ldots,2t-2\}$, so that $|\alpha_i| \leq \eta$ by the approximation-oracle guarantee. We can rewrite $H_t(\tilde{f}, \vec{u})$ (by \Cref{srt_defi:Hankel-matrix-poly}) as:

\begin{align*}
    H_{t}(\Tilde{f},\bm u)\!&=\!\!\begin{pmatrix}
	        \Tilde{f}(\bm u^0) & \!\Tilde{f}(\bm u^1)\! & \!\!\hdots\!\! & \Tilde{f}(\bm u^{t-1})\\
	        \Tilde{f}(\bm u^1) & \!\Tilde{f}(\bm u^2)\! & \!\!\hdots\!\! & \Tilde{f}(\bm u^{t}) \\
	        \vdots & \!\vdots\! & \!\!\ddots\!\! & \vdots \\
	        \Tilde{f}(\bm u^{t-1}) & \!\Tilde{f}(\bm u^{t})\! & \!\!\hdots\!\! & \Tilde{f}(\bm u^{2t-2})
	 \end{pmatrix}\!\!= \!\!\begin{pmatrix}
	        f(\bm u^0)+\alpha_0 & \!\!f(\bm u^1)+\alpha_1 & \!\!\hdots\!\! & f(\bm u^{t-1})+\alpha_{t-1} \\
	        f(\bm u^{1})+\alpha_{1} & \!\!f(\bm u^{2})+\alpha_{2} & \!\!\hdots\!\! & f(\bm u^{t})+\alpha_{t}\\
	        \vdots & \!\!\vdots & \!\!\ddots\!\! & \vdots\\
	        f(\bm u^{t-1})+\alpha_{t-1} & \!\!f(\bm u^{t})+\alpha_{t} & \!\!\hdots\!\! & f(\bm u^{2t-2})+\alpha_{2t-2}
	 \end{pmatrix} \\
    & = \underbrace{\begin{pmatrix}
		        f(\bm u^0) & f(\bm u^1) & \hdots & f(\bm u^{t-1}) \\
		        f(\bm u^1) & f(\bm u^2) & \hdots & f(\bm u^{t}) \\
		        \vdots & \vdots & \ddots & \vdots \\
		        f(\bm u^{t-1}) & f(\bm u^{t}) & \hdots & f(\bm u^{2t-2})
	  \end{pmatrix}}_{=H_{t}(f,\bm u)}
     +\underbrace{\begin{pmatrix}
		        \alpha_0 & \alpha_1 & \hdots & \alpha_{t-1} \\
		        \alpha_1 & \alpha_2 & \hdots & \alpha_{t}\\
		        \vdots & \vdots & \ddots & \vdots \\
		        \alpha_{t-1} & \alpha_{t} & \hdots & \alpha_{2t-2}
		    \end{pmatrix}}_{\triangleq E_t(\bm u)}.
\end{align*}
The Hankel structure of $E_t(\vec{u})$ is evident. $\|E_t(\vec{u})\|_\infty \leq \eta$ follows from the approximation-oracle guarantee. This in turn shows $\|E_t(\vec{u})\|_{\rm F} \leq \sqrt{\eta^2 t^2} = \eta t$, which implies the bound on $\|\E_t(\vec{u})\|_{\rm op}$.
\end{proof}

\end{document}